\documentclass[12pt]{article}

\usepackage[utf8]{inputenc}
\usepackage{amsmath,amsfonts,amsthm,amssymb}
\usepackage{mathtools}
\usepackage{graphicx}
\usepackage{microtype}
\usepackage[english]{babel}
\usepackage{bbm}
\usepackage{natbib} 
\usepackage[psdextra,pdfencoding=auto]{hyperref}

\newcommand{\V}{\mathcal{V}}
\newcommand{\eqdef}{\coloneqq}
\newcommand{\A}{\mathcal{A}}
\renewcommand{\L}{\mathcal{L}}
\newcommand{\N}{\mathcal{N}}
\newcommand{\R}{\mathbb{R}}
\newcommand{\F}{\mathcal{F}}
\renewcommand{\S}{\mathcal{S}}
\newcommand{\E}{\mathbb{E}}
\renewcommand{\P}{\mathbb{P}}
\newcommand{\uv}{\underline v}
\newcommand\given[1][]{%
\nonscript\:#1\vert
\allowbreak
\nonscript\:
\mathopen{}}
\newcommand{\tu}{\tilde u}
\newcommand{\tp}{\tilde p}
\newcommand{\tpi}{\tilde \pi}
\usepackage{cleveref}
\newtheorem{assumption}{Assumption}
\newtheorem{theorem}{Theorem}[section]

\newtheorem{lemma}[theorem]{Lemma}

\newtheorem{proposition}[theorem]{Proposition}

\theoremstyle{definition}
\newtheorem{definition}[theorem]{Definition}
\theoremstyle{remark}

\newtheorem{example}[theorem]{Example}

\RequirePackage{xcolor}
\definecolor{DarkGreen}{rgb}{0.15,0.5,0.15}
\definecolor{DarkRed}{rgb}{0.6,0.2,0.2}
\definecolor{DarkBlue}{rgb}{0.15,0.15,0.55}
\definecolor{DarkPurple}{rgb}{0.4,0.2,0.4}
\hypersetup{
	pdfencoding=unicode,
	linktocpage=true,
	colorlinks=true,                
	linkcolor=DarkBlue,     
	citecolor=DarkBlue, 
	urlcolor=DarkBlue,      
}

\usepackage{mathtools}
\usepackage{tikz}
\usepackage{sgame}
\DeclarePairedDelimiter\set\{\}

\newcommand{\ind}{\mathbf{1}}

\usepackage{pgfplots}
\pgfplotsset{compat=1.17}
\tikzset{>=stealth}
\tikzstyle{info}=[ultra thin,black!15]
\usepgfplotslibrary{colorbrewer}

\usepackage{xcolor}
\usepackage[margin=1in]{geometry}
\usepackage{booktabs} 
\usepackage{subcaption}
\usepackage[onehalfspacing]{setspace}
\newcommand{\pmax}{P_{\max}}
\makeatletter
\def\and{
  \end{tabular}%
  \hskip 0.3em \@plus.17fil%
  \begin{tabular}[t]{c}}
\makeatother

\title{Selling Information in Competitive Environments\thanks{Bonatti acknowledges financial support through NSF Grant SES\textendash1948692. We thank Maryann Rui and Stephen Morris for helpful comments, and seminar participants at UChicago, UIUC, INFORMS, Harvard, and MIT.}}
\author{Alessandro Bonatti\thanks{Massachusetts Institute of Technology, Sloan
School of Management} \and Munther Dahleh\thanks{Massachusetts Institute of
Technology, Institute for Data, Systems, and Society} \and Thibaut
Horel\footnotemark[2] \and Amir Nouripour\footnotemark[2]}
\date{}

\begin{document}

\maketitle

\begin{abstract}
Data buyers compete in a game of incomplete information about which a single data seller owns some payoff-relevant information. The seller faces a joint information- and mechanism-design problem: deciding which information to sell, while eliciting the buyers' types and imposing payments. We derive the welfare- and revenue-optimal mechanisms for a class of games with binary actions and states. Our results highlight the critical properties of selling information in competitive environments: (i) the negative externalities arising from buyer competition increase the profitability of recommending the correct action to one buyer exclusively; (ii) for the buyers to follow the seller's recommendations, the degree of exclusivity must be limited; (iii) the buyers' obedience  constraints also limit the distortions in the allocation of information introduced by a monopolist seller; (iv) as competition becomes fiercer, these limitations become more severe, weakening the impact of market power on the allocation of information.\bigskip

\bigskip

\bigskip
\end{abstract}

\newpage

\section{Introduction}

Markets for information shape a growing fraction of the economy. Information is sold directly (e.g., credit bureaus sell consumer scores to lenders, and research institutions sell data to financial traders) but also indirectly (e.g., digital platforms offer advertisers access to a targeted audience, and hedge funds sell shares of the portfolios they build based on superior information).\footnote{We describe the targeted advertising example at length below. See \citet{adpf90} and \citet{BB19} for a discussion of direct vs.\ indirect sales of information.} The allocation of information affects the distribution of market power in the downstream markets where that information is used, thereby critically impacting consumer and social surplus. Understanding how these markets work, and what is special about them, is then a first-order economic and social issue.

In this paper, we study how private information and buyer competition interact in determining the optimal allocation and price of information. Our objective is threefold: (i)~to provide qualitative insights into the structure of the revenue-maximizing mechanisms for the sale of information, (ii) to determine how information differs from physical goods in this respect, and (iii) to assess the impact of market power in the information sector on competition in downstream markets. 
We propose a tractable formulation of this problem where the profitability of acquiring information for any buyer is unknown to the seller (e.g., buyers have private cost, asset holdings, risk preferences), and buyers of information compete with one another (e.g., financial traders, advertisers, lenders). 
We cast the monopolist seller's choice of a mechanism for the sale of information as an \emph{information design} problem \emph{with elicitation} \citep{BM19}. We consider finitely many data buyers and a data seller. The data buyers compete in a simultaneous-moves, finite game of incomplete information (the ``downstream game''). The monopolist is informed about a payoff-relevant state variable and sells informative signals to the buyers. Each buyer also has a payoff type in the downstream game, i.e., their willingness to pay for any signal is their private information. Thus, the seller must first elicit the buyers' private payoff types and then sell them informative signals. As these signals can be viewed as action recommendations, the seller faces a joint mechanism and information design problem, wherein their choice of information structure is subject to the buyers' obedience and truthful reporting  constraints.

More specifically, a direct mechanism  maps the state of the world and the buyers' reported types into a distribution over informative signals and payments. An important property of our setting is that the seller can design any statistical experiment but lacks complete control over the buyers' actions. This is because information is only valuable insofar as it affects behavior \citep{blackwell53}, and the buyers retain control over their downstream actions. Likewise, the seller can design the information revealed to any buyer  when one or more buyers do not participate in the mechanism. However, the seller cannot prevent any buyer from playing in the downstream game under their prior information only. Therefore, the designer has partial but not full control over each buyer's outside options, which  partially relaxes the buyers' participation constraints.


\paragraph{Main Results} 

We start with a characterization of the seller's constraints in \Cref{sec:ic}. This characterization holds whenever the buyers' payoffs are linear in their private type. For such payoff structures, we show that incentive compatibility of the mechanism is equivalent to \emph{separately} incentivizing truthful reporting and obedience of the buyers. In other words, double deviations—wherein a buyer both misreports their type and deviates from the seller's action recommendation—are no more profitable to the buyers than one-shot deviations.

Next, we focus on the simplest instance of this complex problem---a binary-action downstream game of incomplete information where the state identifies which action is dominant for each data buyer. In this game, acquiring information  imposes negative externalities on the other buyers: the better informed a buyer is about the state, the \emph{lower} the resulting payoff for the other buyers. In other words, the seller designs a  mechanism in the presence of externalities stemming from the competition among buyers \citep{JM06}.



We then turn to the welfare-optimal mechanism for the allocation of information and the revenue-maximizing mechanism for a monopolist seller. Our results highlight two defining features of selling information to competing buyers and show how information and competition interact in shaping the optimal mechanism. Both features distinguish the sale of information to competing firms from the sale of physical goods with externalities across buyers (e.g., network goods).

First, any buyer can always ignore (or indeed reverse) the seller's recommendation. The resulting \emph{obedience} constraint limits the social planner's ability to reveal information to one buyer \emph{exclusively}. Likewise, obedience disciplines the monopolist seller's ability to distort the allocation of information to maximize revenue at the expense of social welfare. Intuitively, the seller wants to distort the allocation of any buyer type with a negative Myersonian virtual value to minimize her payoff and reduce the information rents of higher types. 

In our setting, this distortion corresponds to recommending the wrong action in every state. However, the buyer would not follow such a recommendation in any mechanism that does so too often. Therefore, the seller can do no better than to reveal ``zero net information'' to a low-value buyer, i.e., to probabilistically send the right and the wrong recommendation in a way that leaves the buyer indifferent over any course of action.

There are, of course, many such information structures (including entirely uninformative ones). However, the seller is not indifferent among them. Indeed, she can tailor the joint distribution of recommendations to maximize welfare while maintaining obedience on aggregate. The seller then prefers to reveal the correct state to all buyers when their types are sufficiently low. This approach relaxes obedience constraints and allows the seller to issue the correct recommendation to one or more buyers \emph{exclusively} (and the wrong recommendation to the remaining buyers) when their types are sufficiently larger than their competitors'.

Second, providing information to a firm naturally imposes a negative externality on its competitors. In our setting, these negative externalities expand the profitability of selling information. Each buyer is willing to pay a positive price if either (a) she is strictly better off following the seller's recommendation or (b) her opponents do not receive the correct recommendation with probability one. As a result, the seller can charge a strictly positive price to some types with negative Myersonian virtual values, including some types whose obedience constraint binds, i.e., those who do not receive any valuable information themselves. 


\paragraph{Leading Example} 

Large digital platforms (e.g., Amazon, Facebook, and Google in the US, Alibaba and JD in China) collect an ever-increasing amount of information on their users' online behavior (e.g., browsing, shopping, social media interactions), which allows them to precisely estimate individual consumers' tastes for various products. Our leading application (fleshed out in \Cref{ex:binary}) considers the interaction between a digital platform (information seller) and two or more merchants (information buyers). The merchants wish to leverage the platform's information advantage to offer a personalized  product to each individual consumer.

The platform monetizes its proprietary information by selling  \emph{targeted} advertising slots (e.g., Facebook, Google) or sponsored marketplace listings (e.g., Amazon) to advertisers and retailers. Such practices amount to \emph{indirect sales of information}: while the platform does not trade its consumers' data for payment (\emph{direct sales}), it nonetheless  creates value for merchants by allowing them to condition their strategies (in particular, which product to offer) on the consumers' information (e.g., their browsing or shopping history and third-party cookies). For the purposes of our model, direct transfers of information and indirect sales of information are, in fact, equivalent.\footnote{In our approach, we further assume a platform such as Amazon has full commitment power to set information structures. Recent work by \cite{kosk22} analyzes the problem of information disclosure to multiple agents by a designer without commitment power.}

Each merchant's expected volume of sales depends on two critical factors: (i) the degree of targeting, i.e., the precision of its advertising campaign, as measured by the ability to show the right product to each consumer; and (ii) the exclusivity of its campaign, i.e., the mismatch between its competitors' offers and the consumer's tastes. Merchants are willing to pay more for targeted campaigns and even more for exclusive access to targeted campaigns.

However, the merchants' willingness to pay for an advertising campaign also depends on the profitability of making each sale, i.e., on their cost structure. As the latter is privately known to the merchant, the platform must elicit it through its choice of mechanism. Abstracting from the details and dynamics of online advertising auctions, the platform's problem reduces to designing a menu of (information structure, payment) pairs, each corresponding to an advertising campaign.

\paragraph{Related Literature} This paper is primarily related to the literature on markets for information. In seminal work, \cite{adpf86,adpf90} study the sale of information to traders who compete in financial markets. More recently, \cite{BA12} and \cite{BBS18} study settings where a single buyer has private information about her beliefs over a payoff-relevant state. This problem is similar to ours insofar as the optimal mechanism can be represented through menus of information structures and associated prices, but there is  a  single buyer only.

Closer to our  model, \cite{RO21} studies fully general mechanisms in a model with binary actions, states, and types. However, the buyers' types correspond to the realizations of a privately observed, exogenous signal about the state. \cite{BCT19} also study a setting  similar to ours but consider mechanisms with a single option only---selling the true state distorted by Gaussian noise.  Their problem then consists of finding the optimal covariance matrix of the noise and the associated prices. In particular, the covariance matrix is not designed as a function of the buyers' private types.\footnote{\cite{XS13} study information sellers who offer selling \emph{exogenous} information structures about a binary state to buyers with \emph{known} types who compete in a game with binary actions. \cite{KPP18} study the sale of  cost information to a large number of perfectly competitive firms, each one facing a privately informed manager. \cite{BDW21} consider two sellers who acquire information about a consumer's location along a Hotelling line.}

Our work is also related to the recent literature on Bayesian persuasion and information design, e.g., \cite{BM16}, \cite{KA19}, \cite{T19}, and the references therein. Most papers in this literature view the the problem as a pure information design question as opposed to a mechanism design problem with transfers. In particular, these papers do not study the information structures that maximize the designer's revenue.

In a seminal paper, \cite{jems96} study an auction setting with multidimensional and interdependent valuations. Each buyer is privately informed about her valuation for a good and about the externality that she imposes on others. They show that the revenue-maximizing mechanism may involve not selling the good at all (when this is socially optimal) while  charging positive payments to the buyers.\footnote{\cite{jems99} study the simpler problem where each buyer knows the externality imposed on her by others receiving an object. Under appropriate symmetry assumptions, the problem reduces to a one-dimensional mechanism-design problem. In this vein, \cite{ossa21} study a screening model with externalities where each buyer's type affects both her valuation (e.g., for a network good) and the influence her actions (e.g., consumption) impose on other buyers. Their analysis focuses on the countervailing impact of payoff-types and influence functions.} In further work, \cite{jemo00} restrict attention to the second-price auction and study more general externalities in the game played by the buyers. See \cite{JM06} for an exhaustive survey of the literature on mechanism design with externalities.

Our analysis is also very closely related to the model of data auctions with externalities in \cite{ADHR20}. In their paper, the externalities resulting from the allocation of information are \emph{intrinsic} to the buyers---the negative marginal effect of a competing buyer acquiring information is part of each buyer's  private type. Finally, recent work by \cite{kang22} and \cite{past22} explores a mechanism-design approach to the taxation  of goods with externalities.

Relative to all these papers, our analysis highlights the differences between selling information and traditional goods in markets with \emph{endogenous} downstream externalities. In contrast to the sale of physical goods, the allocation of information is both more flexible and more constrained. On the one hand, the seller has the flexibility to design any statistical experiment for each profile of buyer types. On the other hand, information is an input into the buyers' strategic decision problem (the ``downstream game'') that the seller does not control. As such, the sale of information introduces both obedience constraints, which are new to this literature, and tighter participation constraints.

\section{Model}\label{sec:model}

We consider $n$ data buyers who compete in a downstream game of incomplete information. A monopolist data seller observes a payoff-relevant state variable and  sells informative signals to the data buyers. 

\paragraph{Notation} For a tuple of sets $(\S_i)_{i\in[n]}$, we write
$\S=\prod_{i=1}^n \S_i$ and $\S_{-i}=\prod_{j\neq i} \S_j$. Similarly for $s\in
\S$, $s_i$ (resp. $s_{-i}$) denotes the projection of $s$ on $\S_i$ (resp.
$\S_{-i}$). Finally $\Delta(\S)$ denotes the set of probability distributions
over $\S$.

\paragraph{Downstream Game}
We consider a downstream game of incomplete information between $n$ buyers parametrized by an unknown parameter $\theta$ (the \emph{state of the world}). We denote by $\Theta$ the set of all possible states. Each buyer $i\in[n]$ is described by a set of \emph{types} $\V_i$, a set of actions $\A_i$, and a utility function $$u_i: \A\times\Theta\times\V\to\R.$$

\paragraph{Information Structures} The monopolist information seller  chooses a set of signals $\S$, and a message (bid) space $\mathcal{B}$ to design a communication rule $\sigma:\Theta\times\mathcal{B}\to\Delta(\S)$ and a payment function
$p:\Theta\times\mathcal{B}\to \R_{\geq 0}^n$. Given a vector of bids $b\in\mathcal{B}$ and state $\theta\in\Theta$, we write $\sigma(\,\cdot\,; \theta, b):\S\to[0,1]$ for the
corresponding probability distribution over $\S$.

The buyers' utility functions $(u_i)_{i\in[n]}$, the mechanism $(\sigma,p)$, as well as the joint distribution of the random variables $(\theta, V)\in \Theta\times \V$ are commonly known at the onset of the game.

\paragraph{Timing} The interaction between the information seller and the buyers, and among the buyers in the downstream game, takes place as follows:
\begin{enumerate}
	\item Each buyer $i\in[n]$ observes their type $V_i$ and the  seller observes the state $\theta$.
	\item Each buyer reports a message $B_i$ to the information seller, where $B_i$ is a $V_i$-measurable random variable in $\mathcal{B}_i$.
	\item The information seller generates signals $S\in\S$ distributed as
		$\sigma(\theta, B)$ and reveals $S_i$ to each buyer $i\in[n]$ in exchange for payment $p_i(\theta, B)$.
	\item Each buyer $i$ chooses an action $A_i\in \A_i$ that is $(V_i,
		S_i)$-measurable and obtains a total utility of $u_i(A; \theta, V) - p_i(\theta,
		B)$.
\end{enumerate}

The above formulation reduces the problem of information sale as a joint mechanism and information design problem. In order to obtain closed-form characterizations of the welfare- and revenue-optimal mechanisms in \Cref{sec:lin}, we will restrict our attention to a specific case of this problem that we describe next. However, the characterizations of the  constraints faced by the seller, which we study in \Cref{sec:ic}, hold fore more general games.

\paragraph{Binary Game with Symmetric Additive Payoffs}
We consider a downstream game with $n$ buyers, two states of the world, $\Theta=\{0,1\}$ and two actions for each buyer, $\A_i=\set{0,1}$ for $i\in[n]$. The utility of buyer $i\in[n]$ is given by
\begin{displaymath}
	u_i(a;\theta, v) = v_i\cdot\pi_i(a; \theta).
\end{displaymath}
In other words, the buyers have private  \emph{payoff types} that capture their marginal valuation for the downstream outcomes and reveal nothing about the state of the world.  We further assume that the random variables $(\theta, v_1, \dots, v_n)$ are drawn from a mutually independent prior. In this class of games, we assume the \emph{downstream payoff} $\pi_i$ of buyer $i\in[n]$  is given by
\begin{equation}\label{eq:pi}
	\pi_i(a;\theta) = \ind\set{a_i=\theta} - \frac\alpha{n-1}\sum_{j\neq i}\ind\set{a_j=\theta}\,.
\end{equation}
Thus, in each state of the world, it is a dominant strategy for each buyer to play the action matching that state, resulting in a payoff gain of $1$. A buyer additionally incurs a payoff loss of $\alpha/(n-1)$ whenever one of their competitors plays the action matching the state.\footnote{In the binary game, ``choosing the correct action'' is analogous to ``being awarded the object'' in an auction with externalities. In that case, the function $\pi_i$ corresponds to the value of a given allocation for any buyer $i$, which is then scaled by the buyer's type $v_i$. Throughout the paper, and especially in \Cref{sec:comp}, we will discuss significant differences between the allocation of physical and information goods.} The externalities are thus normalized in such a way that $\alpha$ parametrizes the maximum externality that can be induced on a given buyer by the other buyers.

\begin{example}[Binary Product Choice]\label{ex:binary}
	The binary game described above can be seen as a stylized formulation of the motivating example presented in the introduction, with the state $\Theta=\set{0,1}$ representing an individual consumer's preferences (unknown to the merchants). In this context, the merchants' goal is to match their product to the consumer's preferences. Each merchant is privately informed of its marginal profitability $v_i$ in the downstream market. When there are only two merchants, we can write the
	payoffs \eqref{eq:pi} for each action profile $a\in\{0,1\}^2$ in each state of the world $\theta$ as
\begin{center}
	\renewcommand{\gamestretch}{1.6}
	\gamemathtrue
	\sgcolsep=1em
	\hspace*{\fill}
\begin{game}{2}{2}[$\theta=0$]\label{payoff-matrix}
		& 0      & 1\\
		0   & 1-\alpha,1-\alpha   & 1, -\alpha\\
		1   & \phantom{1}-\alpha,1\phantom{-\alpha}   & 0,\phantom{-}0
\end{game}
\hspace*{3em}
	\begin{game}{2}{2}[$\theta=1$]
	& 0      & 1\\
		0   & 0,\phantom{-}0   & \phantom{1}-\alpha, 1\phantom{-\alpha}\\
	1   & 1,-\alpha   & 1-\alpha, 1-\alpha
\end{game}
\hspace*{\fill}
\end{center}

The parameter $\alpha\geq 0$ captures the intensity of the competition between the merchants. A special case of this game occurs when $\alpha=1$ in which case we have a zero-sum game. For $\alpha>1$, the negative externalities outweigh the direct effect of choosing the correct action: the game turns into a prisoners' dilemma.
Finally, each merchant is privately informed about its unit profit margin, i.e., merchant $i$'s total payoff (gross of any payments to the seller) is given by $v_i\cdot\pi_i$. \end{example}

\section{Incentive Compatibility and Participation}\label{sec:ic}

\subsection{Definitions}\label{sec:def}


\paragraph{Incentive Compatibility} In the game described in Section~\ref{sec:model}, each data buyer makes two strategic decisions: (i) report a message  $B_i\in\mathcal{B}_i$ after observing their private type $V_i$, and (ii) take an action $A_i$ in the downstream game after receiving signal $S_i\in\S_i$ from the  seller.

By the revelation principle for dynamic games, see \citet[Section 6.3]{M91}, it is without loss of generality to assume that the seller's set of signals $\S_i$ is equal to the set of actions $\A_i$, and that the buyers' reports lie in their own type space $\V_i$ instead of a general message space $\mathcal{B}_i$, as long as we consider incentive compatible mechanisms.

In any such mechanism, the seller  recommends to the buyer an action to take in the downstream game. Henceforth, we therefore denote the seller's recommendation by $A_i$, and then buyer's choice of action by $a_i$. Incentive compatibility (below) requires each buyer to both report her true type and to follow the seller's recommendation.

\begin{definition}[Incentive Compatibility]\label{def:ic}
	A mechanism $(\sigma, p)$ is \emph{incentive compatible} if, for each
	$(v_i,v'_i)\in\V_i^2$ and for each deviation function $\delta:\A_i\to\A_i$, 
\begin{multline*}
	\E\big[u_i(A_i, A_{-i}; \theta, V)-p_i(\theta, v_i, V_{-i})\given V_i=v_i, B_i=v_i]
			\geq \\
	\E\big[u_i(\delta(A_i), A_{-i}; \theta, V)-p_i(\theta, v'_i, V_{-i})\given
		V_i=v_i, B_i=v_i'],
\end{multline*}
where $A$ is distributed as $\sigma(\theta, B_i, V_{-i})$. The input into the deviation function $\delta$ is the seller's recommended action, and the output is the buyer's chosen action after receiving the recommendation.
\end{definition}

This definition of incentive compatibility is closely related to the one of \citet[Section 3.1]{BM19} in the context of information design with elicitation (but no transfers). 
 In particular, \Cref{def:ic} requires the mechanism to be robust to \emph{double deviations} in which the data buyer both misreports their private type and deviate from the seller's recommendation. This implies that the mechanism is both \emph{truthful} and \emph{obedient} as defined next.

\begin{definition}[Obedience]\label{def:obedience}
	A mechanism $(\sigma, p)$ is \emph{obedient} if buyers have no incentive to deviate
	from the action recommendation of the  seller assuming everyone
	reports their type truthfully. Formally, for each $\delta:\A_i\to\A_i$ and $v_i\in\V_i$,
	\begin{displaymath}
		\E\big[u_i(A_i, A_{-i}; \theta, V)\given V_i=v_i]
			\geq 
	    \E\big[u_i(\delta(A_i), A_{-i}; \theta, V)\given
		V_i=v_i]
    \end{displaymath}
    where $A$ is distributed as $\sigma(\theta, V)$---in particular, data buyer $i$'s report is truthful.

    Equivalently one can write obedience as:
	for each $(a_i, a_i')\in \A_i^2$
	and $v_i\in\V_i$,
	\begin{displaymath}
		\E\big[u_i(a_i, A_{-i}; \theta, V)\given V_i=v_i, A_i=a_i]
			\geq \E\big[u_i(a_i', A_{-i}; \theta, V)\given V_i=v_i, A_i=a_i]
				\,,
	\end{displaymath}
	where $A$ is distributed as $\sigma(\theta, V)$.
\end{definition}

The first expression shows data buyer $i$'s strategic behavior before receiving the action recommendation when she intends to report her type in the first stage of the game. At this stage, the buyer's strategy specifies a course of action following any action recommendation from the seller. Obedience requires that no deviations  $\delta:\A_i\to\A_i$ are more profitable than obedience, i.e., the identity mapping $id:\A_i\to\A_i$.

The second expression shows data buyer $i$'s strategic behavior after receiving the action recommendation at the second stage, and expresses that no other action results in a better expected utility. As mentioned before, these two are equivalent.

The second expression (which assumes every buyer reports her type truthfully) shows that obedience is only a property of the downstream game and of the recommendation rule $\sigma$, which thus correlates the actions of the data buyers. The distribution of actions resulting from an obedient recommendation rule in a game of incomplete information is a \emph{Bayes correlated equilibrium} as defined and studied in \cite{BM16,BM19}.

\begin{definition}[Truthfulness]\label{def:truthful}
	A mechanism is \emph{truthful} if buyers have no incentive to misreport their
	type, assuming that everyone follows the seller's recommendations in the downstream game. Formally, for each $(v_i,v'_i)\in\V_i^2$,
\begin{multline*}
	\E\big[u_i(A; \theta, V)-p_i(\theta, v_i, V_{-i})\given V_i=v_i, B_i=v_i]\\
		\geq
	\E\big[u_i(A; \theta, V)-p_i(\theta, v_i', V_{-i})\given V_i=v_i, B_i=v_i']
\end{multline*}
where $A$ is distributed as $\sigma(\theta, B_i, V_{-i})$.
\end{definition}

Incentive compatibility implies both obedience and truthfulness, but the converse is not true in general. In Section~\ref{sec:ic-char}, however, we show that with independent, private payoff types and linear valuations, incentive compatibility is equivalent to  obedience and truthfulness.

\paragraph{Participation}  

The data buyers engage in downstream competition even when they acquire no information from the seller. Thus, a complete description of the mechanism must specify the recommendations sent to the participating buyers when one ore more buyers choose not to participate in the mechanism.  In that case, the data seller's recommendations to their competitors affect the non-participating buyers' utilities.

We define each buyer's bid space to be their type space $\V_i$ augmented with the special symbol `$\perp$', representing the decision not to participate. Writing $\mathcal{B}_i\eqdef \V_i\cup\set{\perp}$ for the bid space of buyer $i\in[n]$ and $\mathcal{B}\eqdef\prod_{i=1}^n \mathcal{B}_i$, the communication rule is now a function $\sigma:\Theta\times\mathcal{B}\to\Delta(\A)$ with the constraint that it only sends a recommendation to the participating buyers. In other words,
\begin{displaymath}
	\forall\theta\in\Theta,\;\forall b\in\mathcal{B},\;
	\sigma(\theta,b) \in \Delta(\textstyle\prod_{i:b_i\neq\perp} \A_i).
\end{displaymath}
Similarly, the payment function $p_i:\Theta\times\mathcal{B}\to\R_{\geq 0}$ of buyer $i\in[n]$ satisfies the constraint that $p_i(\theta, b)=0$ whenever $b_i=\perp$.

For each buyer $i\in[n]$, $\sigma$ induces a communication rule $\sigma_{i}^o:\Theta\times\V_{-i}\to\Delta(\A_{-i})$ on the remaining buyers when buyer $i$ chooses not to participate. This induced communication rule is given by,
\begin{displaymath}
	\forall\theta\in\Theta,\;\forall v_{-i}\in\V_{-i},\;
	\sigma_{i}^o(\theta,v_{-i}) \eqdef \sigma(\theta, \perp, v_{-i}).
\end{displaymath}
This communication rule determines the \emph{outside option} available to non-participating buyers: in any equilibrium where every buyer participates, any deviating buyer $i$ chooses her action in the downstream game to be the best response to $\sigma_{i}^o$, resulting in the reservation utility 
\begin{displaymath}
	\max_{a_i\in\A_i}\E\big[u_i(a_i, A_{-i}; \theta,V)\given V_i=v_i]\, ,
\end{displaymath}
where $A_{-i}$ is distributed according to $\sigma_i^o(\theta, V_{-i})$. This is in marked contrast with a monopoly without externalities, in which a non-participating buyer simply receives no allocation, resulting in a vanishing reservation utility. It is also richer than in markets for physical goods with externalities, where a non-participating buyer has no available actions to choose from. We can now state the participation constraint.

\begin{definition}[Individual Rationality]\label{def:ir}
	The mechanism $(\sigma,p)$ is individually rational each for each buyer $i\in[n]$,
    \begin{equation}\label{eq:IR}
                \E\big[u_i(A; \theta, V)-p_i(\theta, V)\given V_i=v_i]\geq
					\max_{a_i\in\A_i}\E\big[u_i(a_i, A_{-i}; \theta,V)\given V_i=v_i]\, ,
    \end{equation}
where $A_{-i}$ is distributed according to $\sigma_i^o(\theta, V_{-i})$.
\end{definition}

Intuitively,  it is always in the seller's interest to relax this constraint as much as possible by selecting the outside communication rule $\sigma_i^o$ that minimizes the right hand side in \eqref{eq:IR}. In other words, the seller “punishes” a non-participating buyer by sending optimal recommendations to the remaining buyers so as to maximize the externalities induced on the deviating buyer. The specific way to achieve this depends on the downstream game and will be made explicit in \Cref{sec:welfare}.


\subsection{Characterizations}
\label{sec:ic-char}

In \Cref{sec:lin}, we shall focus on the binary game with symmetric additive payoffs described in \Cref{sec:model} and solve for the welfare- and revenue optimal mechanisms subject to the incentive compatibility and participation constraints defined in the previous section. To this end, this section provides characterizations of these two constraints  under different combinations of assumptions. These assumptions admit the aforementioned game as a special case but hold  more broadly.

\paragraph{Incentive Compatibility} We begin the analysis with a characterization of incentive compatibility
(Definition~\ref{def:ic}). We first restrict the buyers' utility to be multiplicatively separable in their independent private types and the outcome of the downstream game. In other words, the utility of buyer $i$  depends linearly on their  own type $v_i$, which is a buyer's marginal valuation for their \emph{downstream payoffs} $\pi_i$.


\begin{assumption}[Linear Payoffs]\label{ass:lin}
The random variables $(\theta, V_1,\dots, V_n)$ are mutually independent. For
each $i\in[n]$, $V_i$ is supported on the non-negative reals and the utility of
buyer $i$ is given by 
\begin{displaymath}
	u_i(a;\theta, v) = v_i\cdot\pi_i(a; \theta)
\end{displaymath}
for some downstream payoff function $\pi_i:\A\times\Theta\to\R$. 
\end{assumption}
		
As discussed above, incentive
compatibility rules out double deviations and implies both truthfulness and obedience. Proposition \ref{prop:IC characterization} shows that, under \Cref{ass:lin}, the converse is true and incentive compatibility reduces to requiring truthfulness and obedience separately. In other words, double deviations are not profitable  whenever a mechanism is immune to single deviations.

\begin{proposition}[IC Characterization]\label{prop:IC characterization}
	Under Assumption~\ref{ass:lin}, a mechanism is incentive compatible whenever it is truthful and obedient.
\end{proposition}

\begin{proof}See \Cref{sec:app-char}.\end{proof}



\paragraph{Truthfulness} In order to characterize truthful mechanisms, we follow the classical result of \cite{M81}, which we restate in Proposition \ref{prop:expected-payment} below using our notation. In particular, recall that, in the general formulation of our model, the payment function $p$ is allowed to depend on the realized  state and on every buyer's reported type. However, note that  we can restrict attention without loss of generality to interim-stage payments, i.e., where buyers' payments are a function of their own report only.

Thus, let $(\sigma, p)$ be a mechanism and define for buyer $i\in[n]$, the 	interim downstream payoff $\tpi_i(V_i)\eqdef \E[\pi_i(A; \theta)\given V_i]$ and interim payment $\tp_i(V_i)\eqdef \E[p_i(\theta,V)\given V_i]$. We then have the following familiar characterization result.

\begin{proposition}[Truthfulness Characterization]\label{prop:expected-payment}
	The mechanism $(\sigma, p)$ is truthful if and only if for each buyer $i$:
	\begin{enumerate}
		\item The interim downstream payoff $\tpi_i$ is non-decreasing.
		\item The interim payment $\tp_i$ is given for $v_i\in\V_i$ by
			\begin{equation}\label{eq:payment}
				\tp_i(v_i) = v_i\cdot\tpi_i(v_i) - \uv\cdot\tpi_i(\uv)
				+ \tp_i(\uv) - \int_{\uv}^{v_i} \tpi_i(s)ds\,.
			\end{equation}
	\end{enumerate}
\end{proposition}

\begin{proof}See \Cref{sec:app-char}.\end{proof}

\paragraph{Obedience} Next, we turn to a characterization of obedience for a class of games which includes the binary game with additive payoffs \eqref{eq:pi} but contains more generally all games for which the externalities imposed by buyers $-i$ on buyer $i$ are independent of this buyer's  strategy. In particular, the externalities are not restricted to being additively decomposable.

\begin{assumption}[Binary State, Separable Externalities]\label{ass:ext}
For each buyer $i\in[n]$, the strategy space satisfies $\A_i=\Theta=\set{0,1}$, and the downstream payoff function $\pi_i$ is given by
	\begin{displaymath}
		\pi_i(a;\theta) = \ind\set{a_i=\theta} + E_i(a_{-i};\theta)
	\end{displaymath}
	for some externality function $E_i:\A_{-i}\times\Theta\to\R$.
\end{assumption}

In a game satisfying \Cref{ass:ext}, the dominant strategy for each buyer in the absence of any signal about $\theta$ is to play the action corresponding to the most likely state under the prior. By construction, this  is the  correct action with probability
\begin{equation}\label{eq:pmax}
    \P[A_i=\theta\given V_i] = \max_{k\in\Theta}\P[\theta=k]=:\pmax.
\end{equation}
The characterization of obedience in \Cref{lemm:char} below requires that following the recommended action makes a buyer more likely to be correct than if choosing an action under the common prior.

	\begin{proposition}[Obedience Characterization]\label{lemm:char}
		Under Assumptions~\ref{ass:lin} and \ref{ass:ext}, a recommendation rule is obedient if and only if for each $i\in[n]$, it holds almost surely that
	\begin{displaymath}
		\P[A_i= \theta\given V_i]\geq \pmax.
	\end{displaymath}
\end{proposition}

\begin{proof}See \Cref{sec:app-char}.\end{proof}

In our characterization of optimal mechanisms below, we exploit the strength of this result, i.e., that obedience is a property of the marginal distribution of actions recommended to buyer $i$. In other words,  the designer can flexibly correlate the buyers' actions and state, provided each buyer is recommended the right action often enough \emph{on average}.

\section{Optimal Mechanisms}\label{sec:lin}

We now turn social welfare and revenue maximization. We show below that, for the binary downstream game with additive payoffs in Eq.~\eqref{eq:pi}, both objectives can be written as a weighted sum of the probabilities that the mechanism recommends the dominant strategy to each buyer (see Eq.~\eqref{eq:meta} below). Hence, we first describe in \Cref{sec:meta} an optimal mechanism for a general class of objective functions of this form, which we then instantiate in \Cref{sec:welfare} and \Cref{sec:revenue} to derive mechanisms  maximizing social welfare and revenue, respectively.

\subsection{Optimal Mechanisms}\label{sec:meta}

We consider a general objective function of the form
\begin{equation}\label{eq:meta}
	W\eqdef\E\left[\sum_{i=1}^n w_i(V)\ind\set{A_i=\theta}\right]
	=\sum_{i=1}^n \E\Big[w_i(V)\P[A_i=\theta\given V]\Big]
\end{equation}
for weight functions $w_i:\V\to\R$.


Expression \ref{eq:meta} and the characterization of obedience obtained in \Cref{lemm:char} suggest a convenient parametrization of the seller's problem in terms of the functions $h_i:\V\to[0,1]$ given by $h_i(V)\eqdef \P[A_i=\theta\given V]$ for each player $i\in[n]$. These functions can easily be expressed in terms of the recommendation rule $\sigma$.  Indeed, we have almost surely
\begin{align*}
	\P[A_i=\theta\given V] &= \E[\ind\set{A_i=\theta}\given V]\\
						   &=\sum_{\substack{a\in\A\\a_i=\theta}}\E[\ind\set{A_1=a_1,\dots,A_n=a_n}\given V]
						   =\sum_{\substack{a\in\A\\a_i=\theta}}\E[\sigma(a;\theta, V)\given V].
\end{align*}
Conversely, Lemma \ref{lemm:param} in \Cref{sec:app-meta} shows that it is possible  to construct a recommendation rule which has $h_i$ as its marginals. In other words, any choice of the marginal functions $h_i$ can be ``realized'' by a recommendation rule. Hence, as long as designer's objective and the constraints on the recommendation rule can be expressed in terms of $h_i(V)$, we will directly optimize over these quantities. An optimal information structure $\sigma$ in this class can then be obtained using \Cref{lemm:param}. 

We now  describe a general recommendation rule that optimizes criteria of the form \eqref{eq:meta}, which include social welfare and seller revenue, subject to the obedience constraints. Recall the definition of $\pmax$ given in \eqref{eq:pmax}.


\begin{proposition}[Optimal Mechanism]\label{prop:master}
Under Assumptions~\ref{ass:lin} and \ref{ass:ext}, consider an objective $W$ of the form \eqref{eq:meta} where for $i\in[n]$, $w_i:\V\to\R$ is a measurable function such that the random variable $w_i(v_i, V_{-i})$ is non-atomic for each $v_i\in\V_i$. For $i\in[n]$ let $t_i^\star:\V_i\to\R$ be such that for all $v_i\in\V_i$,
\begin{displaymath}
	\P\big[w_i(v_i, V_{-i})\geq t_i^*(v_i)\big]=\pmax.
\end{displaymath}
Then the deterministic recommendation rule given by
\begin{displaymath}
A_i=\theta\quad\text{if and only if}\quad
w_i(v)\geq \min\set{0,t_i^\star(v_i)}
\end{displaymath}
for $i\in[n]$, maximizes $W$ subject to obedience.
\end{proposition}
\begin{proof}See \Cref{sec:app-meta}.\end{proof}

In order to gain intuition into the characterization of optimal mechanisms, note that the objective function $W$ in \eqref{eq:meta} and the obedience constraints are separable. In other words,  the optimization problem reduces to solving separately for each $i\in[n]$ and $v_i\in\V_i$:
\begin{align*}
	\max&\;\E\big[w_i(v_i, V_{-i})h_i(v_i, V_{-i})\big]\\
	\text{s.t.}&\;\E[h_i(v_i, V_{-i})]\geq \pmax,
\end{align*}
where, as above, $h_i(v)=\P[A_i=\theta\given V]$ is the “allocation of correct information” to buyer $i$ and takes values in $[0,1]$ by definition. In the absence of the obedience constraint, the optimal solution would simply be to choose $h_i(v) = \ind\set{w_i(v)\geq 0}$. If this violates the obedience constraint, we must also  allocate information to some types where $w_i(v)<0$, but we want to do so where the weight function $w_i$ is as large as possible. Hence we should consider the smallest possible superlevel set of $w_i$ that guarantees that the constraint is satisfied. This set corresponds to the level $t_i^\star(v_i)$ defined in the proposition statement.

\subsection{Welfare Maximization}\label{sec:welfare}

We now leverage \Cref{prop:master} to characterize the welfare-optimal mechanism in our environment. For the binary game with additive payoffs \eqref{eq:pi}, we can write the expected social welfare as
\begin{equation}\label{eq:welfare}
\begin{split}
	W &= \sum_{i=1}^n\E\bigg[V_i\Big(\ind\set{A_i=\theta} - \frac\alpha{n-1}\sum_{j\neq i}\ind\set{A_j=\theta}\Big)\bigg]\\
	  &=\sum_{i=1}^n\E\bigg[\Big(V_i-\frac\alpha{n-1}\sum_{j\neq i} V_j\Big)\ind\set{A_i=\theta}\bigg].
\end{split}
\end{equation}
Using the characterization of obedience from \Cref{lemm:char}, the problem of
maximizing social welfare subject to obedience can be written
\begin{displaymath}
	\begin{split}
		\max&\;\sum_{i=1}^n\E\bigg[\Big(V_i-\frac\alpha{n-1}\sum_{j\neq i} V_j\Big)\ind\set{A_i=\theta}\bigg]\\
	\text{s.t.}&\;\P[A_i=\theta\given V_i]\geq\pmax,
 \text{ for $i\in[n]$ and a.s.}
	\end{split}
\end{displaymath}
which is of the form \eqref{eq:meta}. We can thus apply \Cref{prop:master} and obtain the following characterization of the welfare-maximizing (second best) mechanism. 

\begin{proposition}[Welfare Optimal Mechanism]\label{prop:second-best}
Consider the binary game with additive payoffs \eqref{eq:pi} under
\Cref{ass:lin}. Further assume that the buyers' types are identically
distributed with absolutely continuous c.d.f.\ $F$ and denote by $F^{(k)}$ the
c.d.f.\  of the sum of $k$ i.i.d.\ variables\footnote{$F^{(k)}$ can be computed
recursively with $F^{(1)} = F$ and $F^{(k+1)} = F^{(k)}\ast f$, where $\ast$
denotes the convolution product and $f$ is the p.d.f.\ associated with $F$.}
distributed according to $F$. 

Define $v^\star\in\R$ such that $$F^{(n-1)}(v^\star)\eqdef \P\bigg[\sum_{j\neq
i} V_j\leq v^\star\bigg] =\pmax.$$
and $\overline\alpha\eqdef\frac\alpha{n-1}$. Then, the recommendation rule maximizing social welfare
subject to obedience is the deterministic rule given by
\begin{displaymath}\label{eq:welfare-alloc}
A_i=\theta\quad\text{if and only if}\quad
\sum_{j\neq i} v_j\leq\max\set{v^\star, v_i/\overline\alpha}.
\end{displaymath}
\end{proposition}

\begin{proof}See \Cref{sec:app-welfare}.\end{proof}

\begin{figure}[!t]
\begin{tikzpicture}[baseline=(current bounding box.west),scale=1]
\draw [<->]  (0,6) node [left] {$v_2$}  -- (0,0) -- (6,0) node [below] {$v_1$};
\draw [info] (0,0) -- (4, 6) node[sloped, pos=0.1, above right, font=\small] {$v_2=v_1/\alpha$};
\draw [info] (0,0) -- (6, 4) node[sloped, pos=0.1, below right, font=\small] {$v_2=\alpha v_1$};
\draw[info] (2.5, 0) -- (2.5, 6);
\draw[info] (0, 2.5) -- (6, 2.5);
\draw [thick] (0,2.5) -- (1.6667, 2.5) -- (4, 6);
\draw [thick] (2.5,0) -- (2.5, 1.6667) -- (6, 4);
\node [left] at (0,2.5) {$v^\star$};
\node [below] at (2.5,0) {$v^\star$};
\node at (1.25,4.5) {$\set{2}$};
\node at (4.5,1.25) {$\set{1}$};
\node at (3.2,3) {$\set{1,2}$};
\node[above] at (current bounding box.north) {$\alpha=2/3$};
\end{tikzpicture}
\hfill
\begin{tikzpicture}[baseline=(current bounding box.west),scale=1]
\draw [<->]  (0,6) node [left] {$v_2$}  -- (0,0) -- (6,0) node [below] {$v_1$};
\draw[info] (0,0) -- (4, 6) node[sloped, pos=0.1, above right, font=\small] {$v_2=\alpha v_1$};
\draw[info] (0,0) -- (6, 4) node[sloped, pos=0.1, below right, font=\small] {$v_2=v_1/\alpha$};
\draw[info] (2.5, 0) -- (2.5, 6);
\draw[info] (0, 2.5) -- (6, 2.5);
\draw [thick] (0,2.5) -- (3.75, 2.5) -- (6, 4); 
\draw [thick] (2.5,0) -- (2.5, 3.75) -- (4, 6);
\node [left] at (0,2.5) {$v^\star$};
\node [below] at (2.5,0) {$v^\star$};
\node at (1.5,4) {$\set{2}$};
\node at (4,1.5) {$\set{1}$};
\node at (4,4) {$\varnothing$};
\node at (1.6,1.6) {$\set{1,2}$};
\node[above] at (current bounding box.north) {$\alpha=3/2$};
\end{tikzpicture}
\caption{Welfare-maximizing recommendation rule from \Cref{prop:second-best} with two buyers, for $\alpha=2/3$ (left) and $\alpha=3/2$ (right). The label in each region indicates the set of buyers who are recommended the correct action ($A_i=\theta$)—buyers in the complement set are recommended the wrong action ($A_i=1-\theta$). The two states are equally likely ex ante, so $v^\star=F^{-1}(1/2)$ is the median of the type distribution—chosen to be a standard exponential here).}
\label{fig:welfare-mech}
\end{figure}
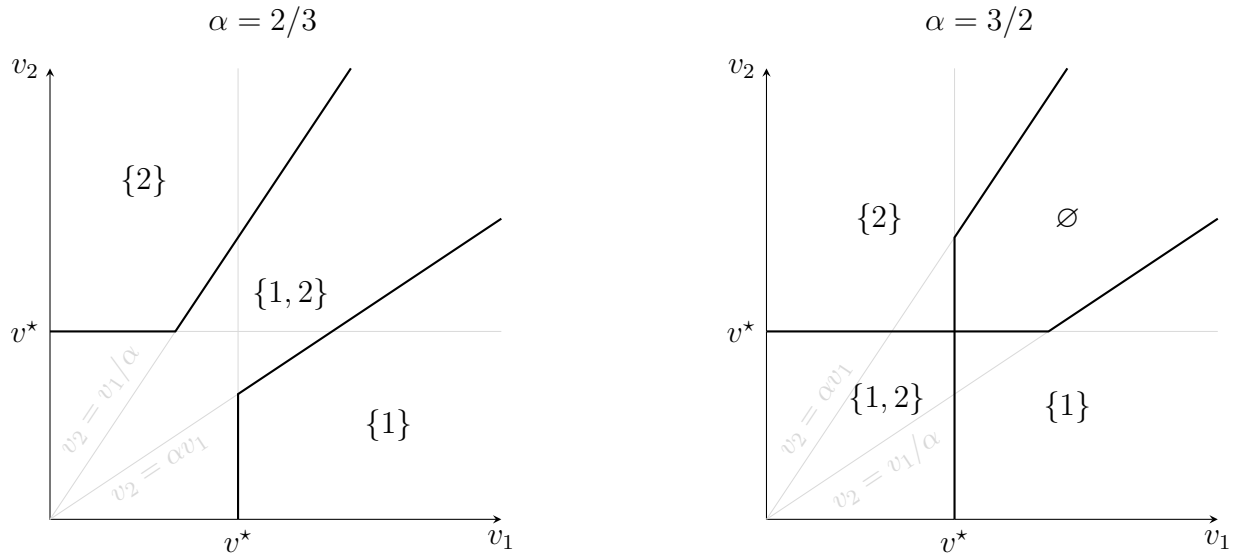

\Cref{fig:welfare-mech} gives a representation of the welfare-optimal recommendation rule from Proposition~\ref{eq:welfare-alloc} in the two-buyer case. This recommendation can be conceptualized as the “superposition” of two recommendation rules, which we describe separately in \Cref{fig:min-externality}.
\begin{enumerate}
	\item The first rule (\Cref{fig:min-externality}, left) recommends the correct action to buyer $i$ if and only if buyer $j$'s type satisfies $v_j\leq v^\star$. For this rule, the recommendation to buyer $i$ is independent or their type and satisfies $\P[A_i=\theta\given V_i] = F(v^\star) = \pmax$. In other words, the recommendation is correct as often as buyer $i$ would be by deterministically playing the action matching the most likely state under the prior. This implies by the characterization of \Cref{lemm:char} that the mechanism is obedient. Consequently, this mechanism recommends the correct action to buyer $i$ \emph{just often enough} to ensure obedience, and does so when buyer $j$'s type is lowest, thus minimizing the induced externality $\alpha v_j\ind\set{A_i=\theta}$. In summary, this mechanism ensures each buyer's  obedience while minimizing the externality induced on the other buyer. Note that the mechanism is obedient despite the action recommendations being deterministic in each region. This is because from the perspective of each buyer, conditional on their type, the recommendation they receive is still a random variable depending on the (unobserved) realization of the other buyer's type.
	\item The second rule (\Cref{fig:min-externality}, center and right) recommends the correct action to buyer $i$ if and only if her type satisfies $v_i\geq \alpha v_j$. This is simply the welfare-maximizing allocation (in the absence of the obedience constraint): a buyer is recommended the correct action if her value exceeds the externality she imposes on the other buyer. In particular, when buyer $i$'s type is large enough compared to buyer $j$'s type ($v_i/v_j\geq \max\set{\alpha,1/\alpha}$), she is recommended the right action exclusively, hence maximizing her utility. In the intermediate region where types are close to each other, both buyers are recommended the same action. When $\alpha\leq1$, the region is defined by $\alpha v_j\leq v_i\leq v_j/\alpha$ and the efficient allocation recommends the correct action to both buyers. In contrast, when $\alpha> 1$, the region is defined by $v_j/\alpha\leq v_i\leq \alpha v_j$ and both buyers are recommended the wrong action. Indeed, the externalities are so significant in this case that the buyers face a prisoners' dilemma in each state. It is thus more efficient for the data seller to coordinate the buyers on the collaborative strategy in which both buyers pick the “wrong” action.
\end{enumerate}

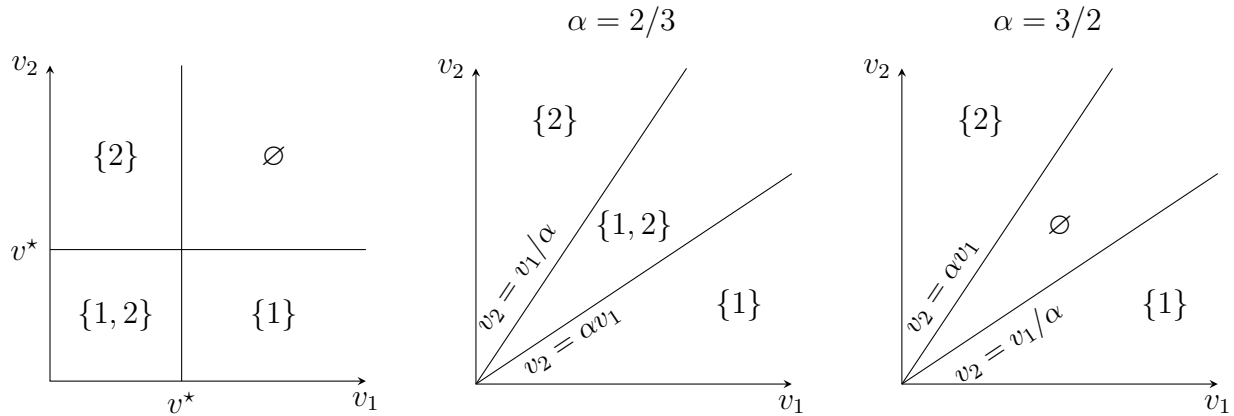
\begin{figure}[!t]
	\centering
	\begin{tikzpicture}[baseline=(current bounding box.south),scale=0.7]
\draw [<->]  (0,6) node [left] {$v_2$}  -- (0,0) -- (6,0) node [below] {$v_1$};
\draw (2.5, 0) node[below] {$v^\star$} -- (2.5, 6);
\draw (0, 2.5) node [left] {$v^\star$} -- (6, 2.5);
\node at (1.25,4.25) {$\set{2}$};
\node at (4.25,1.25) {$\set{1}$};
\node at (1.25,1.25) {$\set{1,2}$};
\node at (4.25,4.25) {$\varnothing$};
\end{tikzpicture}\hfill
	\begin{tikzpicture}[baseline=(current bounding box.south),scale=0.7]
\draw [<->]  (0,6) node [left] {$v_2$}  -- (0,0) -- (6,0) node [below] {$v_1$};
\draw (0,0) -- (4, 6) node[sloped, pos=0.1, above right, font=\small] {$v_2=v_1/\alpha$};
\draw (0,0) -- (6, 4) node[sloped, pos=0.1, below right, font=\small] {$v_2=\alpha v_1$};
\node at (1.5,5) {$\set{2}$};
\node at (5,1.5) {$\set{1}$};
\node at (3,3) {$\set{1,2}$};
\node[above] at (current bounding box.north) {$\alpha=2/3$};
\end{tikzpicture}\hfill
\begin{tikzpicture}[baseline=(current bounding box.south),scale=0.7]
\draw [<->]  (0,6) node [left] {$v_2$}  -- (0,0) -- (6,0) node [below] {$v_1$};
\draw (0,0) -- (4, 6) node[sloped, pos=0.1, above right, font=\small] {$v_2=\alpha v_1$};
\draw (0,0) -- (6, 4) node[sloped, pos=0.1, below right, font=\small] {$v_2=v_1/\alpha$};
\node at (1.5,5) {$\set{2}$};
\node at (5,1.5) {$\set{1}$};
\node at (3,3) {$\varnothing$};
\node[above] at (current bounding box.north) {$\alpha=3/2$};
\end{tikzpicture}
\caption{Building blocks for the welfare-maximizing mechanism of \Cref{prop:second-best}. Left: mechanism guaranteeing obedience at all types while minimizing externalities. Center and right: first-best mechanism (ignoring the obedience constraint) for $\alpha=2/3$ and $\alpha=3/2$.}
\label{fig:min-externality}
\end{figure}

The optimal mechanism (\Cref{fig:welfare-mech}) combines both mechanisms by
distorting the first best mechanism to guarantee that each buyer $i$ receives
the correct action when $v_j\leq v^\star$. Distorting buyer
$i$'s recommendation is required, and hence obedience is binding, when
$v_j\leq v^\star$ and $v_i\leq \alpha v_j$.

Finally, it is easy to verify that the second best mechanism is implementable, i.e., it satisfies the buyers' truth-telling constraints. Indeed, by \Cref{prop:expected-payment} it suffices to verify that the interim downstream payoff is non-decreasing in the buyer's type.

\begin{proposition}[Implementability of Second-Best Mechanism]\label{prop:implementability}
	For the deterministic mechanism of \Cref{prop:second-best}, and under the same assumptions, the interim expected payoff of buyer $i\in[n]$,
	$\tpi_i(V_i)\eqdef\E\big[\pi_i(A;\theta)\given V_i\big]$, satisfies almost
	surely
	\begin{displaymath}
		\tpi_i(v_i) = \max\set[\big]{F^{(n-1)}(v^\star),F^{(n-1)}(v_i/\overline\alpha)}
		-\overline\alpha\sum_{j\neq i}
		\E\Big[F^{(n-2)}\big(\max\set{v^\star,V_j/\overline\alpha}-v_i\big)\Big].
	\end{displaymath}
	In particular, $\tpi_i$ is non-decreasing and the recommendation rule is
	therefore implementable.
\end{proposition}
\begin{proof}See \Cref{sec:app-welfare}.\end{proof}


Intuitively, a higher type is revealed the correct state more often by the social planner, which makes it possible to find transfers that would induce truthful reporting of the buyers' types. Of course these transfers do not correspond to a monopolist data seller's optimal choice. In the next section, we will see how a monopolist data seller modifies the second best mechanism to maximize the associated payments.
 
\subsection{Revenue Maximization}\label{sec:revenue}

Throughout this section, we further assume that the type distribution $F$ is
absolutely continuous with p.d.f.\ $f$ and that the \emph{virtual value
function} $\phi:\V_i\to\R$ defined by
\begin{displaymath}
	\phi(v)\eqdef v - \frac{1-F(v)}{f(v)},
\end{displaymath}
is non-decreasing, that is, $F$ is \emph{regular} in the sense of \cite{M81}.

We first show in Lemma~\ref{lem:vv} that maximizing the seller's expected
revenue reduces to maximizing the virtual
surplus, as in \cite{M81}. 

\begin{lemma}[Reduction to Virtual Surplus]\label{lem:vv}

Let $\sigma$ be a communication rule for which the interim payoff $\tpi_i$ is
non-decreasing for each buyer $i\in[n]$. Denote by $K$ the interim downstream
payoff of a non-participating buyer\footnote{Using the notations of
	\Cref{def:ir}, if $\sigma_i^o$ denotes the recommendation rule used with
	the remaining buyers when buyer $i$ does not participate, then we have
$K=\E[\pi_i(a^*, A_{-i};\theta)]$, where $A_{-i}$ is distributed according to
$\sigma_i^o(\theta,V_{-i})$ and $a^*$ is the action matching the most likely
state under the prior.} and assume that $\tpi_i(\uv)\geq K$. Then:
	\begin{enumerate}
		\item If $p_i$ is a payment function that truthfully implements
			$\tpi_i$ (i.e., that satisfies \eqref{eq:payment}
			by \Cref{prop:expected-payment}), then $(\sigma,p)$ is individually
			rational if and only if it is individually rational for the lowest
			type, that is, $\tp_i(\uv)\leq \uv\cdot (\tpi_i(\uv) - K)$.
		\item Among the payment functions $p_i$ implementing $\tpi_i$ in
			a truthful and individually rational manner, the revenue-maximizing
			one is given by
		\begin{equation}\label{eq:payments-bis}
			\tp_i(v_i) = v_i\cdot\tpi_i(v_i) - \uv\cdot K - \int_{\uv}^{v_i}\tpi_i(s)ds\,.
		\end{equation}
		For this payment function, the seller's revenue is $R=
		\sum_{i=1}^n\E\big[\phi(V_i)\tpi_i(V_i)\big]-n\uv\cdot K$.
	\end{enumerate}
\end{lemma}

\begin{proof}See \Cref{sec:app-revenue}.\end{proof}

We thus focus on maximizing the virtual surplus $R^\dagger\eqdef\sum_{i\in\set{1,2}}\E\big[\phi(V_i)\tpi_i(V_i)\big]$ subject to obedience and truthfulness. For the binary game~\eqref{eq:pi}, we write the virtual surplus as
\begin{displaymath}
	R^\dagger
	=\sum_{i=1}^n\E\bigg[\Big(\phi(V_i)-\frac\alpha{n-1}\sum_{j\neq i}\phi(V_j)\Big)\ind\set{A_i=\theta}\bigg].
\end{displaymath}
This objective function is of the form~\eqref{eq:meta} and we can thus apply \Cref{prop:master} to characterize the communication rule maximizing virtual surplus subject to obedience. Then, we verify that the corresponding expected  downstream payoff, $\tpi_i$, is non-decreasing, implying that the mechanism is implementable in a truthful and individually rational manner using the payments given by \eqref{eq:payments-bis}.

\begin{proposition}[Revenue Optimal Mechanism]\label{prop:revenue-alpha}

Consider the binary game with additive payoffs~\eqref{eq:pi} under \Cref{ass:lin}. Further assume that the buyers' types are identically distributed with absolutely continuous c.d.f.\ $F$. Denote by $F_\phi$ the c.d.f.\footnote{When $F$ is a regular distribution, the virtual value function $\phi$ is invertible and the c.d.f.\ $F_\phi$ can be computed as $F\circ\phi^{-1}$.}\ of $\phi(V_i)$ where $V_i$ is distributed according to $F$ and by $F_\phi^{(k)}$ the c.d.f.\ of the sum of $k$ i.i.d.\ variables distributed according to $F_\phi$.

Define $v^\star$ such that $F_\phi^{(n-1)}\big(\phi(v^\star)\big) = \pmax$ 
and $\overline\alpha\eqdef \alpha/(n-1)$. Then, the recommendation rule maximizing virtual surplus subject to obedience is the deterministic rule given by
\begin{displaymath}
A_i=\theta\quad\text{if and only if}\quad
\sum_{j\neq i} \phi(v_j)\leq\max\set{\phi(v^\star), \phi(v_i)/\overline\alpha}.
\end{displaymath}
\end{proposition}

\begin{proof}
	The proof is identical to the one of \Cref{prop:second-best} with $\phi(V_i)$
	playing the role of $V_i$. It follows from an application of
	\Cref{prop:master} with weight function $w_i(v)= \phi(v_i) -
	\overline\alpha\sum_{j\neq i}\phi(v_j)$.
\end{proof}

\begin{figure}[tbp]
\begin{tikzpicture}[scale=3,baseline=(current bounding box.west),
	declare function = {
	b2(\x)=(\x+1)/2;}]
	\pgfmathsetmacro{\vstar}{ln(2)}
	\pgfmathsetmacro{\vtilde}{b2(\vstar)}
\draw [<->]  (0,2) node [left] {$v_2$}  -- (0,0) -- (2,0) node [below] {$v_1$};
\draw [info] (1, 0) node [below,black] {$v_0$} -- (1, 1);
\draw [info] (0, 1) node [left,black] {$v_0$} -- (1, 1);
\draw [info] (0, {b2(0)}) -- (2, {b2(2)});
\draw [info] ({b2(0)}, 0) -- ({b2(2)}, 2);
\draw [info] (\vstar, 0) -- (\vstar, 2); 
\draw [info] (0,\vstar) -- (2, \vstar);
\draw [thick] (\vstar,0) node [below] {$v^\star$} -- (\vstar, \vtilde) -- (2, {b2(2)});
\draw [thick] (0,\vstar) node [left] {$v^\star$} -- (\vtilde, \vstar) -- ({b2(2)}, 2);
\node at (0.5, 1.5) {$\set{2}$};
\node at (1.5, 0.5) {$\set{1}$};
\node at (1.5, 1.5) {$\set{1,2}$};
\node at (0.35, 0.35) {$\set{1,2}$};
\node at (0.8, 0.8) {$\varnothing$};
\node[above] at (current bounding box.north) {$\alpha=1/2$};
\end{tikzpicture}\hfill
\begin{tikzpicture}[scale=3,baseline=(current bounding box.west),
	declare function = {
		b2(\x)=2*\x-1; }]
	\pgfmathsetmacro{\vstar}{ln(2)}
	\pgfmathsetmacro{\vtilde}{b2(\vstar)}
\draw [<->]  (0,2) node [left] {$v_2$}  -- (0,0) -- (2,0) node [below] {$v_1$};
\draw [info] (1, 0) node [black,below] {$v_0$} -- (1, 1);
\draw [info] (0, 1) node [black,left] {$v_0$} -- (1, 1);
\draw [info] (0.5, {b2(0.5)}) -- (1.5, {b2(1.5)});
\draw [info] ({b2(0.5)}, 0.5) -- ({b2(1.5)}, 1.5);
\draw [info] (\vstar, 0) -- (\vstar, 2); 
\draw [info] (0,\vstar) -- (2, \vstar);
\draw [thick] (\vstar,0) node [below] {$v^\star$} -- (\vstar, \vtilde) -- (1.5, {b2(1.5)});
\draw [thick] (0,\vstar) node [left] {$v^\star$} -- (\vtilde, \vstar) -- ({b2(1.5)}, 1.5);
\node at (0.5, 1.5) {$\set{2}$};
\node at (1.5, 0.5) {$\set{1}$};
\node at (1.5, 1.5) {$\varnothing$};
\node at (0.4, 0.4) {$\set{1,2}$};
\node[above] at (current bounding box.north) {$\alpha=2$};
\end{tikzpicture}
\caption{Revenue-maximizing recommendation rule from \Cref{prop:revenue-alpha}
for $\alpha=1/2$ (left) and $\alpha=2$ (right). Types are distributed
exponentially, so that $\phi(v)=v-1$ and $v_0=\phi^{-1}(0) = 1$. The prior on
$\theta$ is symmetric ($\pmax=1/2$), hence $v^\star=F^{-1}(1/2)=\ln 2<v_0$.}
\label{fig:blocks-revenue}
\end{figure}
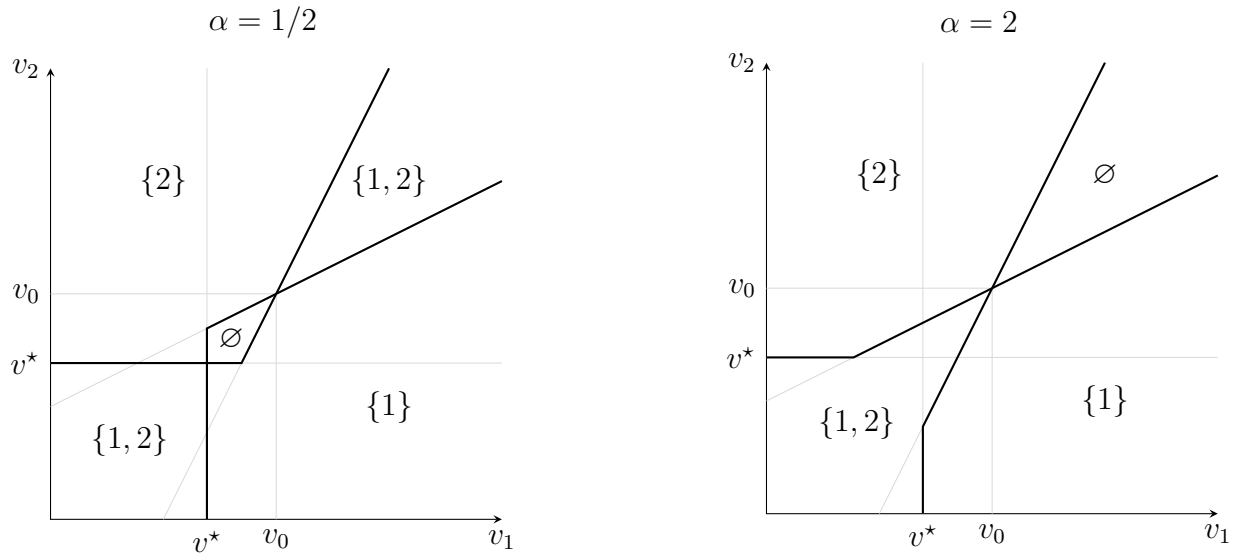


The functional form of the revenue-optimal mechanism in
\Cref{prop:revenue-alpha} is analogous to that of the welfare-optimal mechanism
in \Cref{prop:second-best}, after replacing the buyers' types with their
virtual types. Figure~\ref{fig:blocks-revenue} shows the resulting
recommendation rule for $n=2$ buyers both when $\alpha < 1$ and $\alpha > 1$.
Again, this recommendation can be understood as the superposition of two
recommendation rules:
\begin{itemize}
	\item The first rule recommends the correct action to buyer $i$ if and only if the virtual type of the other buyer satisfies $\phi(v_j)\leq \phi(v^\star)$, or equivalently since $F$ is regular, $v_j\leq v^\star$. This is exactly the same mechanism as was shown in \Cref{fig:min-externality} (left) guaranteeing the obedience of buyer $i$.
	\item The second rule recommends the correct action to buyer $i$ if and only if $\phi(v_i)\leq \phi(v_j)/\alpha$. In particular, when one virtual valuation is large compared to the other ($\phi(v_i)/\phi(v_j)\geq\max\set{\alpha,1/\alpha}$), buyer $i$ is recommended the correct action exclusively. However, because the functions $v\mapsto \phi^{-1}(\phi(v)/\alpha)$ and $v\mapsto\phi^{-1}(\alpha\phi(v))$ intersect at $v_0\eqdef \phi^{-1}(0)$, the intermediate regime $\phi(v_i)/\phi(v_j)<\max\set{\alpha,1/\alpha}$ now determines two regions in which both buyers are recommended the same action. When virtual valuations are positive (types greater than $v_0$), both buyers are recommended the correct action when $\alpha < 1$ and the wrong action when $\alpha > 1$. Indeed, in this latter case, the buyers face a prisoners' dilemma in which coordinating on the dominated “wrong” action results in higher payoffs. Naturally, the situation is reversed when virtual values are negative in the intermediate regime: both buyers receive the wrong action when $\alpha<1$ and the correct one when $\alpha>1$. This is shown in \Cref{fig:surplus-alpha} below.
\end{itemize}
The revenue-optimal mechanism resulting from the superposition of these two mechanisms depends both qualitatively and quantitatively on the relative positions of $v_0$ and $v^\star$. This in turn depends on the magnitude of the parameter $\pmax$ and is discussed in \Cref{sec:prior} below.

\Cref{prop:rev-implementability} below gives an expression for the expected downstream payoff $\tpi_i$ of each buyer in the obedient mechanism described above, and states that it is non-decreasing. Consequently, the mechanism above is also \emph{truthful (implementable)} and the payments are then given by \Cref{lem:vv}. Given that these payments are decreasing as a function on $K$, the downstream payoff of a non-participating buyer, we must therefore design the outside option so as to minimize $K$.
For the binary game with additive payoffs, the following proposition
establishes that the optimal allocation when $i$ does not participate recommends the correct action to the set $[n]\setminus\set{i}$ of all participating buyers.

\begin{proposition}[Implementation of Optimal Mechanism]\label{prop:rev-implementability}
For the mechanism of \Cref{prop:revenue-alpha} and under the same assumptions, the interim downstream payoff $\tpi_i$ of buyer $i\in[n]$ is the non-decreasing function
\begin{multline*}
\tpi_i(v_i) =
\max\set[\big]{F_\phi^{(n-1)}\big(\phi(v^\star)\big),F_\phi^{(n-1)}(\phi(v_i)/\overline\alpha)}\\
-\overline\alpha\sum_{j\neq i}
\E\Big[F_\phi^{(n-2)}\big(\max\set[\big]{\phi(v^\star),\phi(V_j)/\overline\alpha}-\phi(v_i)\big)\Big],
\end{multline*}
and the revenue-maximizing mechanism is therefore implementable in a truthful manner.

In case of non-participation of buyer $i\in[n]$, the recommendation rule minimizing their reservation utility recommends the correct action to the remaining buyers ($A_j=\theta$ for $j\neq i$). For this outside option, the payments maximizing revenue subject to individual rationality and truthfulness are given by
\begin{displaymath}
\tp_i(v_i) = v_i\cdot\tpi_i(v_i) - \int_{\uv}^{v_i}\tpi_i(s)ds
+\uv\alpha -\uv\cdot\pmax.
\end{displaymath}
\end{proposition}

\begin{proof}See \Cref{sec:app-revenue}.\end{proof}

For the outside option in \Cref{prop:rev-implementability}, the optimal strategy of a non-participating buyer is simply to play the action matching the most likely state under the prior, resulting in the buyer being correct with probability $\pmax$. Furthermore, the externality incurred by a non-participating buyer is $(n-1)\overline\alpha=\alpha$, because all participating buyers receive the correct action recommendation in this case. Hence, the reservation utility of a non-participating buyer is $\pmax-\alpha$: this is precisely the offset appearing in the expression for $\tp_i$, in \Cref{prop:rev-implementability} guaranteeing buyer $i$'s participation.

We now remark on several striking properties of the optimal payments, which apply whenever $v^\star<v_0$, as in \Cref{fig:blocks-revenue}. 

\begin{enumerate}
\item Unlike in settings without externalities, merely having a negative virtual value does not imply a buyer receives no information. Even absent obedience constraints, the seller knows that distorting one buyer's recommendation increases the surplus of the other buyer. Therefore, when $\alpha<1$ both buyers receive the wrong recommendation only if both their virtual values are negative \emph{and} they are sufficiently similar.  Conversely, if both virtual values are negative but $v_1$ is sufficiently larger than $v_2$, then the seller prefers issuing the correct recommendation to buyer $1$. Indeed, distorting the recommendation to buyer $1$ would increase buyer $2$'s payoff, which has an even stronger negative impact on the seller's profits. 

\item Some types of buyer $i$ with a negative virtual valuation $v_i<v_0$, 
are nonetheless charged a positive payment. This occurs  because these types are sufficiently high that their opponent $j$ has an even lower type $v_j$ with a significant probability, $F(v_i)$. In other words, the seller finds it optimal to reveal the correct state to buyer $i$ with probability, $F(\phi^{-1}(\phi(v_i)/\alpha))>F(v^\star)$. Buyer $i$ then has a strict incentive to follow the seller's recommendation, i.e., her obedience constraint is slack.

\item Some types of buyer $i$ such that $v^\star<v_i$, 
 whose obedience constraint binds, still pay a strictly  positive price. Because their obedience constraint is binding, these types derive no net utility from following the seller's recommendation. However, unlike types in $[0,v^\star]$ where the other data buyer always receives the right recommendation, these types' opponent is revealed the correct state with probability $1-F(\phi^{-1}(\alpha\phi(v_i)))$. These types are strictly better off participating, and they can be charged a positive payment. Thus, the presence of negative externalities augments the  profitability of selling information, as the seller charges positive payments in exchange for limiting the information available to each buyer's competitors.
\end{enumerate}

\section{Information and Competition}\label{sec:comp}

In this section, we discuss the impact of the environment facing the buyers on the optimal mechanisms presented in \Cref{sec:lin}. This encompasses both the information structure and in particular the buyers' prior information discussed in \Cref{sec:prior} as well as the competition structure in the downstream game as captured by the externality parameter $\alpha$ (\Cref{sec:externality}) and the number of buyers (\Cref{sec:nplayers}).

\subsection{Buyers’ Prior Information}\label{sec:prior}

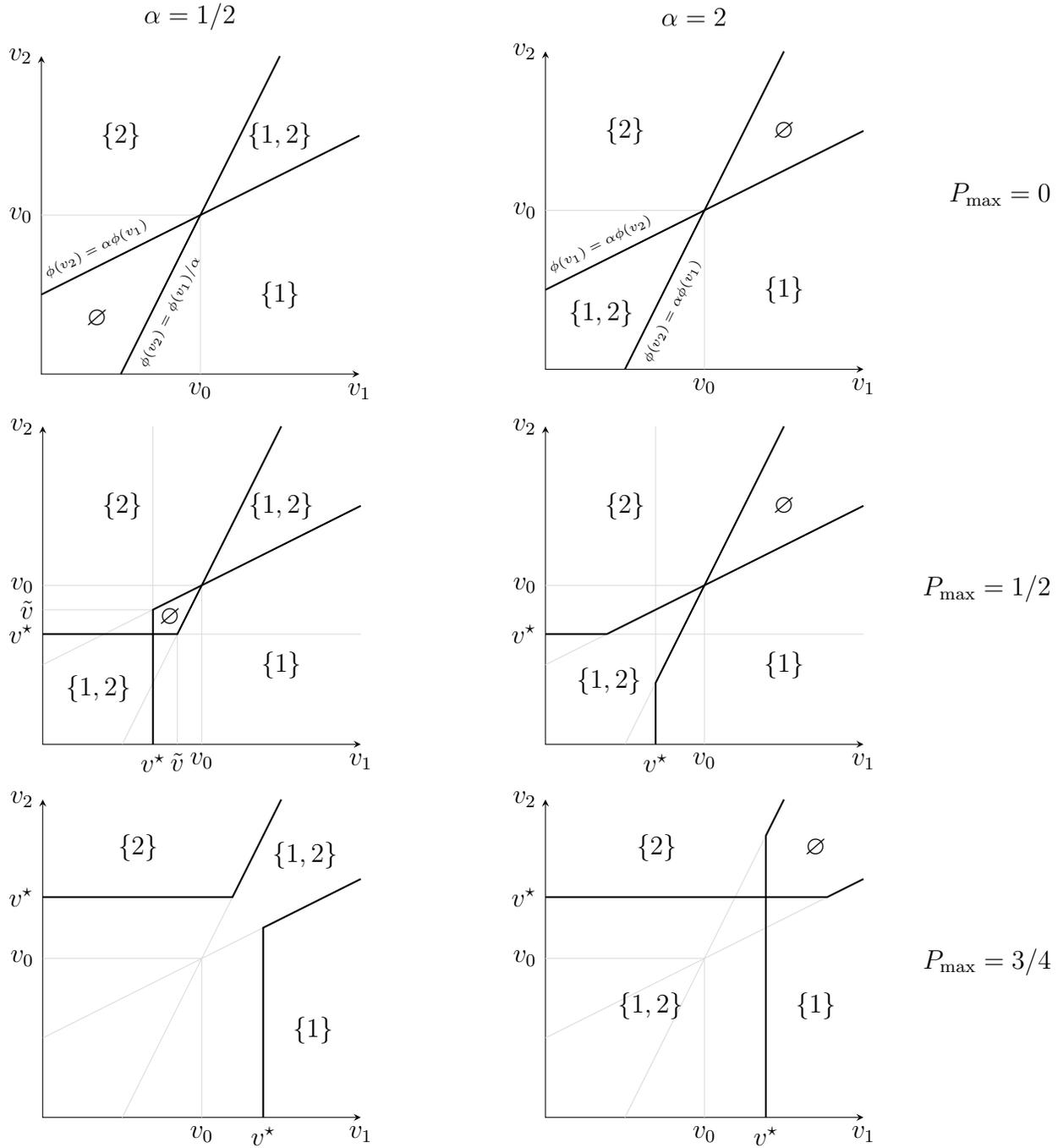
\begin{figure}[!htp]
	\begin{tikzpicture}[baseline=(current bounding box.west),scale=2.5]
\draw [<->]  (0,2) node [left] {$v_2$}  -- (0,0) -- (2,0) node [below] {$v_1$};
\draw [info] (1, 0) node [below, color=black] {$v_0$} -- (1, 1);
\draw [info] (0, 1) node [left, color=black] {$v_0$} -- (1, 1);
\draw [thick] (0, 0.5) -- (2, 1.5) node [above right, pos=0.01, sloped, font=\tiny] {$\phi(v_2)=\alpha\phi(v_1)$} ;
\draw [thick] (0.5, 0) -- (1.5, 2) node[below right, pos=0.02, sloped, font=\tiny] {$\phi(v_2)=\phi(v_1)/\alpha$};
\node at (0.5, 1.5) {$\set{2}$};
\node at (1.5, 0.5) {$\set{1}$};
\node at (1.5, 1.5) {$\set{1,2}$};
\node at (0.35, 0.35) {$\varnothing$};
\node[above] at (current bounding box.north) {$\alpha=1/2$};
\end{tikzpicture}\hfill
\begin{tikzpicture}[baseline=(current bounding box.west),scale=2.5]
\draw [<->]  (0,2) node [left] {$v_2$}  -- (0,0) -- (2,0) node [below] {$v_1$};
\draw [info] (1, 0) node [below, color=black] {$v_0$} -- (1, 1);
\draw [info] (0, 1) node [left, color=black] {$v_0$} -- (1, 1);
\draw [thick] (0, 0.5) -- (2, 1.5) node [above right, pos=0.01, sloped, font=\tiny] {$\phi(v_1)=\alpha\phi(v_2)$} ;
\draw [thick] (0.5, 0) -- (1.5, 2) node[below right, pos=0.02, sloped, font=\tiny] {$\phi(v_2)=\alpha\phi(v_1)$};
\node at (0.5, 1.5) {$\set{2}$};
\node at (1.5, 0.5) {$\set{1}$};
\node at (1.5, 1.5) {$\varnothing$};
\node at (0.35, 0.35) {$\set{1,2}$};
\node[above] at (current bounding box.north) {$\alpha=2$};
\end{tikzpicture}
\begin{minipage}{6em}
	\begin{flushright}

$\pmax=0$
	\end{flushright}
\end{minipage}

\begin{tikzpicture}[scale=2.5,baseline=(current bounding box.west),
	declare function = {
	b2(\x)=(\x+1)/2;}]
	\pgfmathsetmacro{\vstar}{ln(2)}
	\pgfmathsetmacro{\vtilde}{b2(\vstar)}
\draw [<->]  (0,2) node [left] {$v_2$}  -- (0,0) -- (2,0) node [below] {$v_1$};
\draw [info] (1, 0) node [below,black] {$v_0$} -- (1, 1);
\draw [info] (0, 1) node [left,black] {$v_0$} -- (1, 1);
\draw [info] (\vtilde, 0) node [below,black] {$\tilde{v}$} -- (\vtilde, \vstar);
\draw [info] (0, \vtilde) node [left,black] {$\tilde{v}$} -- (\vstar, \vtilde);
\draw [info] (0, {b2(0)}) -- (2, {b2(2)});
\draw [info] ({b2(0)}, 0) -- ({b2(2)}, 2);
\draw [info] (\vstar, 0) -- (\vstar, 2); 
\draw [info] (0,\vstar) -- (2, \vstar);
\draw [thick] (\vstar,0) node [below] {$v^\star$} -- (\vstar, \vtilde) -- (2, {b2(2)});
\draw [thick] (0,\vstar) node [left] {$v^\star$} -- (\vtilde, \vstar) -- ({b2(2)}, 2);
\node at (0.5, 1.5) {$\set{2}$};
\node at (1.5, 0.5) {$\set{1}$};
\node at (1.5, 1.5) {$\set{1,2}$};
\node at (0.35, 0.35) {$\set{1,2}$};
\node at (0.8, 0.8) {$\varnothing$};
\end{tikzpicture}\hfill
\begin{tikzpicture}[scale=2.5,baseline=(current bounding box.west),
	declare function = {
		b2(\x)=2*\x-1; }]
	\pgfmathsetmacro{\vstar}{ln(2)}
	\pgfmathsetmacro{\vtilde}{b2(\vstar)}
\draw [<->]  (0,2) node [left] {$v_2$}  -- (0,0) -- (2,0) node [below] {$v_1$};
\draw [info] (1, 0) node [black,below] {$v_0$} -- (1, 1);
\draw [info] (0, 1) node [black,left] {$v_0$} -- (1, 1);
\draw [info] (0.5, {b2(0.5)}) -- (1.5, {b2(1.5)});
\draw [info] ({b2(0.5)}, 0.5) -- ({b2(1.5)}, 1.5);
\draw [info] (\vstar, 0) -- (\vstar, 2); 
\draw [info] (0,\vstar) -- (2, \vstar);
\draw [thick] (\vstar,0) node [below] {$v^\star$} -- (\vstar, \vtilde) -- (1.5, {b2(1.5)});
\draw [thick] (0,\vstar) node [left] {$v^\star$} -- (\vtilde, \vstar) -- ({b2(1.5)}, 1.5);
\node at (0.5, 1.5) {$\set{2}$};
\node at (1.5, 0.5) {$\set{1}$};
\node at (1.5, 1.5) {$\varnothing$};
\node at (0.4, 0.4) {$\set{1,2}$};
\end{tikzpicture}
\begin{minipage}{6em}
	\begin{flushright}
		$\pmax=1/2$
	\end{flushright}
\end{minipage}

\begin{tikzpicture}[scale=2.5,baseline=(current bounding box.west),
	declare function = {
	b2(\x)=(\x+1)/2;}]
	\pgfmathsetmacro{\vstar}{2*ln(2)}
	\pgfmathsetmacro{\vtilde}{b2(\vstar)}
\draw [<->]  (0,2) node [left] {$v_2$}  -- (0,0) -- (2,0) node [below] {$v_1$};
\draw [info] (1, 0) node [black,below] {$v_0$} -- (1, 1);
\draw [info] (0, 1) node [black,left] {$v_0$} -- (1, 1);
\draw [info] (0, {b2(0)}) -- (2, {b2(2)});
\draw [info] ({b2(0)}, 0) -- ({b2(2)}, 2);
\draw [thick] (\vstar,0) node [below] {$v^\star$} -- (\vstar, \vtilde) -- (2, {b2(2)});
\draw [thick] (0,\vstar) node [left] {$v^\star$} -- (\vtilde, \vstar) -- ({b2(2)}, 2);
\node at (0.6, 1.7) {$\set{2}$};
\node at (1.7, 0.55) {$\set{1}$};
\node at (1.65, 1.65) {$\set{1,2}$};
\end{tikzpicture}\hfill
\begin{tikzpicture}[scale=2.5,baseline=(current bounding box.west),
	declare function = {
		b2(\x)=2*\x-1;}]
		\pgfmathsetmacro{\vstar}{2*ln(2)}
		\pgfmathsetmacro{\vtilde}{b2(\vstar)}
\draw [<->]  (0,2) node [left] {$v_2$}  -- (0,0) -- (2,0) node [below] {$v_1$};
\draw [info] (1, 0) node [black,below] {$v_0$} -- (1, 1);
\draw [info] (0, 1) node [black,left] {$v_0$} -- (1, 1);
\draw [info] (0.5, {b2(0.5)}) -- (1.5, {b2(1.5)});
\draw [info] ({b2(0.5)}, 0.5) -- ({b2(1.5)}, 1.5);
\draw [thick] (\vstar,0) node [below] {$v^\star$} -- (\vstar, \vtilde) -- (1.5, {b2(1.5)});
\draw [thick] (0,\vstar) node [left] {$v^\star$} -- (\vtilde, \vstar) -- ({b2(1.5)}, 1.5);
\node at (0.7, 1.7) {$\set{2}$};
\node at (1.7, 0.7) {$\set{1}$};
\node at (1.7, 1.7) {$\varnothing$};
\node at (0.66, 0.7) {$\set{1,2}$};
\end{tikzpicture}
\begin{minipage}{6em}
	\begin{flushright}
		$\pmax =3/4$
	\end{flushright}
\end{minipage}
\caption{Revenue-maximizing recommendation rule from \Cref{prop:revenue-alpha}
for $\alpha=1/2$ (left) and $\alpha=2$ (right). Types are distributed
exponentially, so that $\phi(v) = v-1$ and $v_0=\phi^{-1}(0)=1$. The first row
shows the first best mechanism. The second row is the second-best mechanism
(subject to obedience) with a symmetric prior on $\theta$, for which $v^\star
= F^{-1}(1/2)=\ln 2<v_0$. The third row is the second-best mechanism with an
asymmetric prior ($p_{\max}=3/4$), for which $v^\star=\ln 4> v_0$.}
\label{fig:surplus-alpha}
\end{figure}

The seller's information augments the buyers' prior information and allows them to tailor their actions to the state of the world. But each buyer also has the option of playing the downstream game under their prior information only. Thus, each buyer's participation constraint is tighter when buyers are better informed, and the seller cannot extract all the buyers' surplus through transfers. 

Furthermore,  while the seller is unconstrained in her choice of experiments, the buyers retain the flexibility to choose their actions after observing the signals. These signals must then be sufficiently informative \emph{relative to the buyers' prior} in order for the data buyers to follow them. In particular, the buyers can always ignore the recommendation altogether and choose the action that is optimal under the prior, or choose actions that respond to signals in a different way than the seller intended. Thus, not all distributions over action profiles in the downstream game are feasible for the seller, due to the buyers' obedience constraints. 

In our binary setting, the seller's problem therefore depends critically on a scalar parameter: the informativeness of the buyers' prior beliefs, as captured by $\pmax\eqdef\max_{k\in\set{0,1}}\P[\theta=k]$. Indeed, $\pmax$ describes both novel aspects of our  seller's problem: the buyer's reservation utility that corresponds to foregoing participation in the mechanism (in which case  the seller fully reveals the state to the $n-1$ other buyers); and the buyer's option value of participating but ignoring the seller's recommendations. Formally, the obedience and participation constraints can be written as 
\begin{align*}
	\P[A_i=\theta\given V_i=v_i]&\geq \pmax, \\
    v_i\tpi_i(v_i)-p_i(v_i)&\geq v_i\left(\pmax-\alpha\right).
\end{align*}

\Cref{fig:surplus-alpha} illustrates the revenue-maximizing mechanism for $n=2$ buyers as we vary the parameter $\pmax$, for both $\alpha<1$ and $\alpha>1$. The top row describes the benchmark case where we artificially set $\pmax=0$ in the constraints above. This corresponds to a setting akin to the sale of physical goods with externalities: there are no downstream actions for buyers to take (hence  no obedience constraints), and each buyer's reservation utility consists of not receiving the good while their competitors all receive it with probability $1$, so that $\tpi=-\alpha$.

As $\pmax$ increases, as in the second row, the revenue-maximizing mechanism must assign the correct action to both buyers whenever \emph{their competitor's} type is below the critical level $v^\star:=F^{-1}(\pmax)$. When $v^\star\geq v_0$ (as in the third row), the square where $v_1,v_2\leq v^\star$, in which obedience requires recommending the optimal action to both buyers, fully contains the region with
negative virtual valuations and the situation looks qualitatively the same as
\Cref{fig:welfare-mech}.\footnote{Depending on the type distribution, we may have $v^\star\geq v_0$ for any prior on the unknown state. For example, for types uniformly distributed over $[0,1]$, $\phi(v) = v-1$ and $v_0=\phi^{-1}(0)= 1/2\leq \pmax=v^\star$.}
In all cases obedience binds for buyer $i$ for all types such that $\phi(v_i)\leq\alpha\phi(v^\star)$, or equivalently $v_i\leq\tilde{v}\eqdef\phi^{-1}(\alpha \phi(v^\star))$.

Finally, \Cref{fig:mech-payment-pmax} shows the payments under the revenue-optimal mechanism for several values of $\pmax$. As previewed in \Cref{sec:revenue}, the model without obedience constraints $(\pmax=0)$ has positive payments for almost all types. In contrast, the precision of the buyers' prior information (captured by $\pmax\geq 1/2$) prevents the seller from charging any payments to a positive measure of types, despite the presence negative externalities. This is in sharp contrast with the sale of \emph{final goods} with externalities, e.g., \cite{jems96}.

\begin{figure}
\noindent\begin{tikzpicture}[baseline]
	\begin{axis}[
		xlabel=type,
		ylabel=payment,
		width=0.5\textwidth,
		cycle list/Dark2,
		legend cell align=left,
		legend entries={$\pmax=0$, $\pmax=1/2$, $\pmax=3/4$},
		legend pos=south east,
		title={$\alpha=1/2$},
		xmax=3,
	]
		\addplot+[line width=0.75pt,mark=none] table [col sep=comma,x=value,y=pay0.5] {mech-fb-exp-pmax-0.5.dat};
		\addplot+[line width=0.75pt,mark=none] table [col sep=comma,x=value,y=pay0.5] {mech-exp-pmax-0.5.dat};
		\addplot+[line width=0.75pt,mark=none] table [col sep=comma,x=value,y=pay0.5] {mech-exp-pmax-0.75.dat};
	\end{axis}
\end{tikzpicture}\hfill\begin{tikzpicture}[baseline]
	\begin{axis}[
		xlabel=type,
		ylabel=payment,
		width=0.5\textwidth,
		cycle list/Dark2,
		legend cell align=left,
		legend entries={$\pmax=0$, $\pmax=1/2$, $\pmax=3/4$},
		legend pos=south east,
		title={$\alpha=2$},
		xmax=3,
	]
		\addplot+[line width=0.75pt,mark=none] table [col sep=comma,x=value,y=pay2] {mech-fb-exp-pmax-0.5.dat};
		\addplot+[line width=0.75pt,mark=none] table [col sep=comma,x=value,y=pay2] {mech-exp-pmax-0.5.dat};
		\addplot+[line width=0.75pt,mark=none] table [col sep=comma,x=value,y=pay2] {mech-exp-pmax-0.75.dat};
	\end{axis}
\end{tikzpicture}\hfill
\caption{Payment as a function of a buyer's type, for different values of
$\pmax$ with exponentially distributed types. Left panel: $\alpha=1/2$. Right panel: $\alpha=2$.}
\label{fig:mech-payment-pmax}
\end{figure}
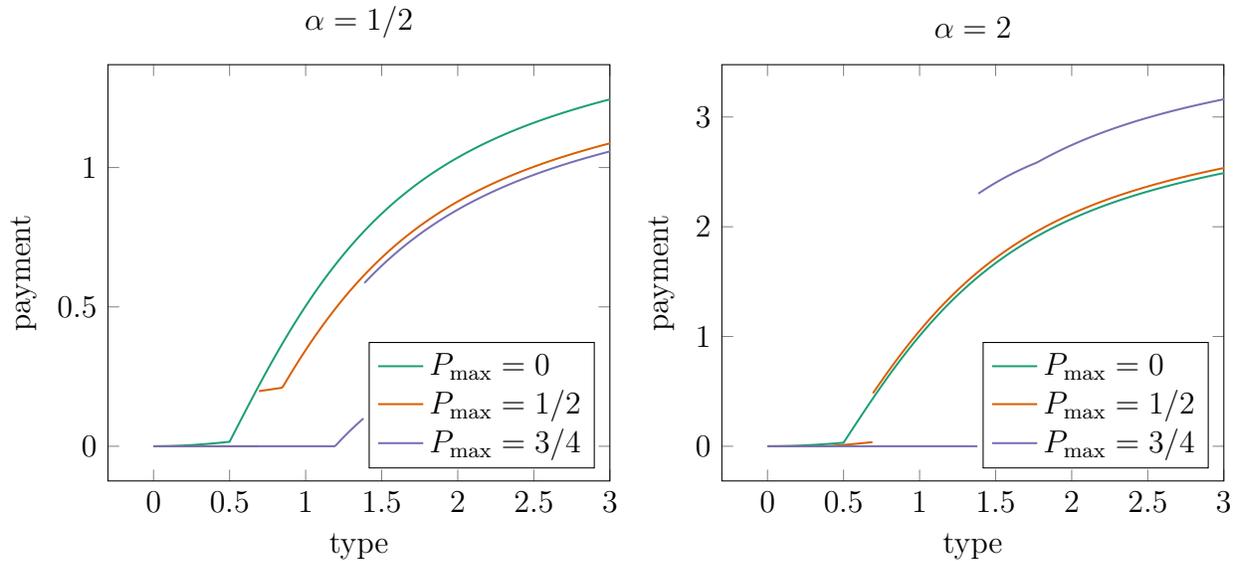

\subsection{Externality Parameter \texorpdfstring{$\alpha$}{α}}\label{sec:externality}
We now investigate the effect of the intensity of downstream competition on the revenue-optimal mechanism. Figure \ref{fig:alpha-comp} compares two settings, where competition is fiercer in the left panel $(\alpha=1/2)$ than in the right panel $(\alpha=1/4)$. 

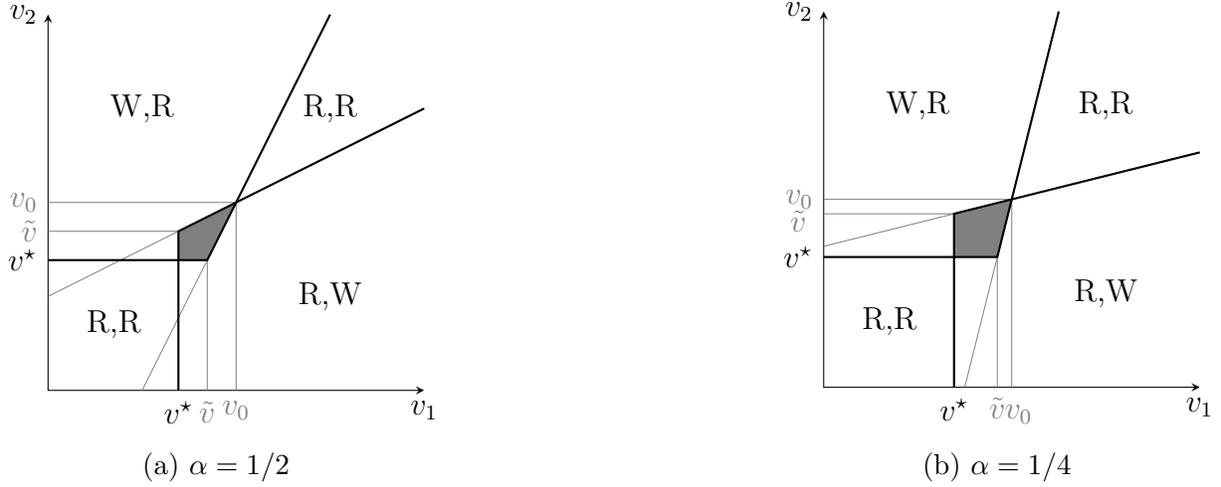
\begin{figure}[!ht]
\begin{center}
	\subcaptionbox{$\alpha=1/2$}{
\begin{tikzpicture}[scale=2.5]
\path [fill=gray] (0.6931,0.6931) -- (0.84655, 0.6931) -- (1,1) -- (0.6931,
0.84655) -- (0.6931, 0.6931);
\draw [<->]  (0,2) node [left] {$v_2$}  -- (0,0) -- (2,0) node [below] {$v_1$};
\draw [help lines] (1, 0) node [below] {$v_0$} -- (1, 1);
\draw [help lines] (0, 1) node [left] {$v_0$} -- (1, 1);
\draw [help lines] (0.84655, 0) node [below] {$\tilde{v}$} -- (0.84655, 0.6931);
\draw [help lines] (0, 0.84655) node [left] {$\tilde{v}$} -- (0.6931, 0.84655);
\draw [thick, help lines] (0, 0.5) -- (2, 1.5);
\draw [thick, help lines] (0.5, 0) -- (1.5, 2);
\draw [thick] (0.6931,0) node [below] {$v^\star$} -- (0.6931, 0.84655) -- (2, 1.5);
\draw [thick] (0,0.6931) node [left] {$v^\star$} -- (0.84655, 0.6931) -- (1.5, 2);
\node at (0.5, 1.5) {W,R};
\node at (1.5, 0.5) {R,W};
\node at (1.5, 1.5) {R,R};
\node at (0.35, 0.35) {R,R};
\end{tikzpicture}
}\hfill\subcaptionbox{$\alpha=1/4$}{\begin{tikzpicture}[scale=2.5]
\path [fill=gray] (0.6931,0.6931) -- (0.9232, 0.6931) -- (1,1) -- (0.6931, 0.9232) -- (0.6931, 0.6931);
\draw [<->]  (0,2) node [left] {$v_2$}  -- (0,0) -- (2,0) node [below] {$v_1$};
\draw [help lines] (1, 0) node [below,yshift=-0.5ex,xshift=+0.5ex] {$v_0$} -- (1, 1);
\draw [help lines] (0, 1) node [left] {$v_0$} -- (1, 1);
\draw [help lines] (0.9232, 0) node [below] {$\tilde{v}$} -- (0.9232, 0.6931);
\draw [help lines] (0, 0.9232) node [left,yshift=-0.5ex,xshift=-0.5ex] {$\tilde{v}$} -- (0.6931, 0.9232);
\draw [thick, help lines] (0, 0.75) -- (2, 1.25);
\draw [thick, help lines] (0.75, 0) -- (1.25, 2);
\draw [thick] (0.6931,0) node [below] {$v^\star$} -- (0.6931, 0.9232) -- (2, 1.25);
\draw [thick] (0,0.6931) node [left] {$v^\star$} -- (0.9232, 0.6931) -- (1.25, 2);
\node at (0.5, 1.5) {W,R};
\node at (1.5, 0.5) {R,W};
\node at (1.5, 1.5) {R,R};
\node at (0.35, 0.35) {R,R};
\end{tikzpicture}}
\caption{Comparison of the revenue-maximizing recommendation rules from \Cref{prop:revenue-alpha} for two different values of $\alpha$. The mechanism recommends the wrong action $(1-\theta)$ to both buyers in the gray regions. As in \Cref{fig:surplus-alpha}, types are exponentially distributed and states are equally likely: the larger the value of $\alpha$, the more competitive the downstream game.}
\label{fig:alpha-comp}
\end{center}
\end{figure}

Reducing the intensity of competition reduces the value of exclusive sales of information (i.e., recommending the right action to one buyer only) in the first best: at one extreme, if buyers imposed no externalities on each other, the seller would recommend the right action to any buyer with a positive virtual value. In particular, in an unconstrained revenue problem, the seller would recommend the wrong action to both buyers more often when competition is weaker. However, this recommendation profile would violate obedience, which requires the seller to recommend the right action to both buyers when both their types are smaller than $v^\star$. As $v^\star$ is independent of $\alpha$, the right panel shows how the seller uses exclusive sales as a second-best policy under obedience constraints more often as the competition weakens.

\begin{figure}[!ht]
\noindent\begin{tikzpicture}[baseline]
	\begin{axis}[
		xlabel=type,
		ylabel=downstream payoff $\tpi_i$,
		width=0.49\textwidth,
		cycle list/Dark2,
		legend cell align=left,
		legend entries={$\alpha=1/2$, $\alpha=1$, $\alpha=3/2$},
		legend pos=south east,
		extra x ticks=0.693,
		extra x tick labels=$v^\star$,
		extra x tick style={grid=major, grid style={solid,gray!30},tick style={draw=none}},
	]
		\addplot+[line width=1pt,mark=none] table [col sep=comma,x=value,y=share0.5] {mech-exp-pmax-0.5.dat};
		\addplot+[line width=1pt,mark=none] table [col sep=comma,x=value,y=share1] {mech-exp-pmax-0.5.dat};
		\addplot+[line width=1pt,mark=none] table [col sep=comma,x=value,y=share1.5] {mech-exp-pmax-0.5.dat};
	\end{axis}
\end{tikzpicture}
\hfill
\begin{tikzpicture}[baseline]
	\begin{axis}[
		xlabel=type,
		ylabel=payment,
		width=0.49\textwidth,
		cycle list/Dark2,
		legend cell align=left,
		legend entries={$\alpha=1/2$, $\alpha=1$, $\alpha=3/2$},
		legend pos=south east,
		extra x ticks=0.693,
		extra x tick labels=$v^\star$,
		extra x tick style={grid=major, grid style={solid,gray!30},tick style={draw=none}},
	]
		\addplot+[line width=1pt,mark=none] table [col sep=comma,x=value,y=pay0.5] {mech-exp-pmax-0.5.dat};
		\addplot+[line width=1pt,mark=none] table [col sep=comma,x=value,y=pay1] {mech-exp-pmax-0.5.dat};
		\addplot+[line width=1pt,mark=none] table [col sep=comma,x=value,y=pay1.5] {mech-exp-pmax-0.5.dat};
	\end{axis}
\end{tikzpicture}
\caption{Interim downstream payoff and payment as functions of a buyer's type in
	the optimal mechanism for different values of $\alpha$. Types are
	exponentially distributed and the prior on the unknown state is uniform
	($\pmax=1/2$). There is a discontinuity at $v^\star=\ln 2<v_0$ and
	a singularity at $\tilde v_\alpha
= \phi^{-1}\big(\alpha\phi(v^\star)\big)$. The minimum expected downstream payoff,
$\tpi_i(\uv)$, is $\pmax -\alpha$.}
\label{fig:opt-mech-uniform}
\end{figure}
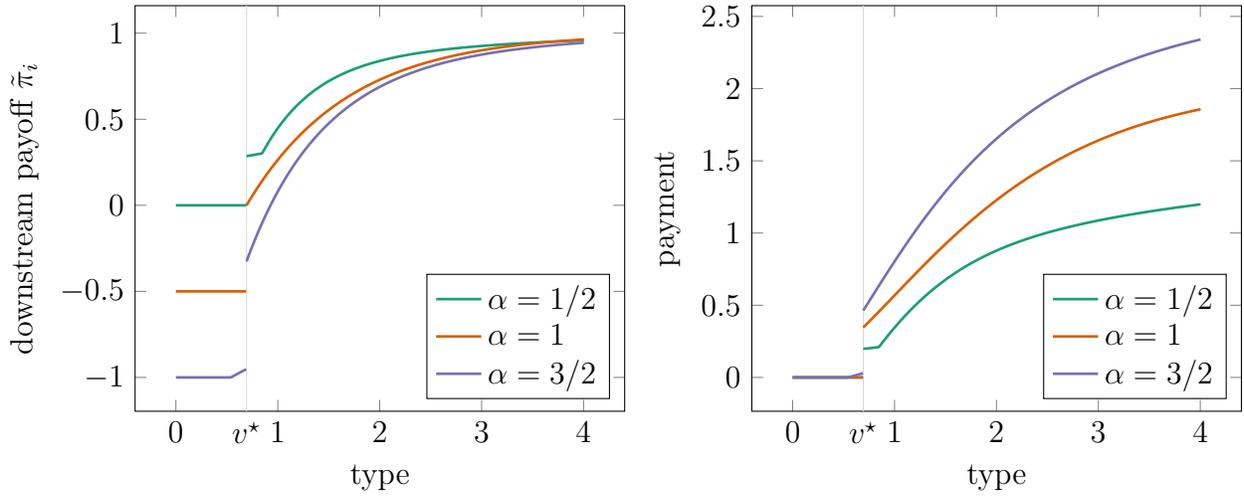

\begin{figure}[!ht]
\noindent\begin{tikzpicture}[baseline]
	\begin{axis}[
		xlabel=type,
		ylabel=downstream payoff $\tpi_i$,
		width=0.5\textwidth,
		cycle list/Dark2,
		legend cell align=left,
		legend entries={$\alpha=1/2$, $\alpha=1$, $\alpha=3/2$},
		legend pos=south east,
		extra x ticks=1.38,
		extra x tick labels=$v^\star$,
		extra x tick style={grid=major, grid style={solid,gray!30},tick style={draw=none}},
	]
		\addplot+[line width=1pt,mark=none] table [col sep=comma,x=value,y=share0.5] {mech-exp-pmax-0.75.dat};
		\addplot+[line width=1pt,mark=none] table [col sep=comma,x=value,y=share1] {mech-exp-pmax-0.75.dat};
		\addplot+[line width=1pt,mark=none] table [col sep=comma,x=value,y=share1.5] {mech-exp-pmax-0.75.dat};
	\end{axis}
\end{tikzpicture}
\hfill
\begin{tikzpicture}[baseline]
	\begin{axis}[
		xlabel=type,
		ylabel=payment,
		width=0.5\textwidth,
		cycle list/Dark2,
		legend cell align=left,
		legend entries={$\alpha=1/2$, $\alpha=1$, $\alpha=3/2$},
		legend pos=south east,
		extra x ticks=1.38,
		extra x tick labels=$v^\star$,
		extra x tick style={grid=major, grid style={solid,gray!30},tick style={draw=none}},
	]
		\addplot+[line width=1pt,mark=none] table [col sep=comma,x=value,y=pay0.5] {mech-exp-pmax-0.75.dat};
		\addplot+[line width=1pt,mark=none] table [col sep=comma,x=value,y=pay1] {mech-exp-pmax-0.75.dat};
		\addplot+[line width=1pt,mark=none] table [col sep=comma,x=value,y=pay1.5] {mech-exp-pmax-0.75.dat};
	\end{axis}
\end{tikzpicture}
\caption{Interim downstream payoff and payment. Same parameters as Figure \ref{fig:opt-mech-uniform} except that now $\pmax=3/4$ so that $v^\star=\ln 4>v_0$.}
\label{fig:opt-mech-nonuniform}
\end{figure}
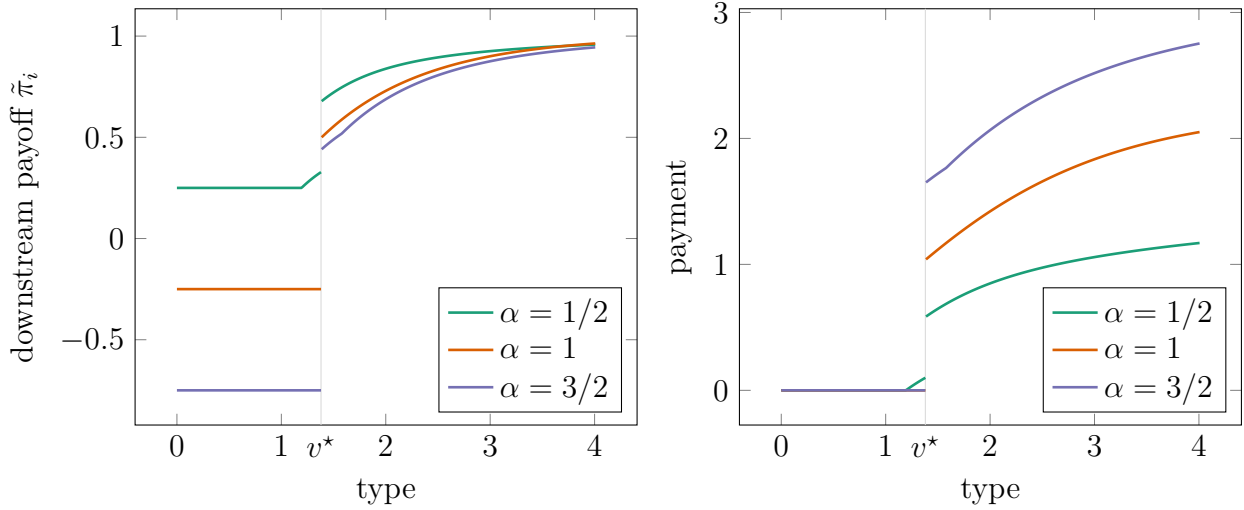

\begin{figure}[!ht]
\begin{tikzpicture}[baseline]
	\begin{axis}[
		xlabel=$\alpha$,
		width=0.54\textwidth,
		cycle list/Dark2,
		legend cell align=left,
		legend entries={$1^{\rm st}$ best,$2^{\rm nd}$ best,Revenue},
		legend pos=north east,
		title={$p_{\rm max}=1/2$},
	]
	\addplot+[line width=1pt,mark=none] table [col sep=comma,x=alpha,y=firstbest] {competition-pmax-point5.dat};
	\addplot+[line width=1pt,mark=none] table [col sep=comma,x=alpha,y=secondbest] {competition-pmax-point5.dat};
	\addplot+[line width=1pt,mark=none] table [col sep=comma,x=alpha,y=revenue] {competition-pmax-point5.dat};
	\end{axis}
\end{tikzpicture}\hfill\begin{tikzpicture}[baseline]
	\begin{axis}[
		xlabel=$\alpha$,
		width=0.54\linewidth,
		cycle list/Dark2,
		legend cell align=left,
		legend entries={$1^{\rm st}$ best,$2^{\rm nd}$ best,Revenue},
		legend pos=north east,
		title={$p_{\rm max}=3/4$},
	]
	\addplot+[line width=1pt,mark=none] table [col sep=comma,x=alpha,y=firstbest] {competition-pmax-point75.dat};
	\addplot+[line width=1pt,mark=none] table [col
	sep=comma,x=alpha,y=secondbest] {competition-pmax-point75.dat};
	\addplot+[line width=1pt,mark=none] table [col sep=comma,x=alpha,y=revenue]
	{competition-pmax-point75.dat};
	\end{axis}
\end{tikzpicture}
\caption{Welfare in the first best and second best mechanisms and revenue in
the optimal mechanism as a function of $\alpha$, for two different priors on the
unknown state.}
\label{fig:wel/rev-alpha}
\end{figure}
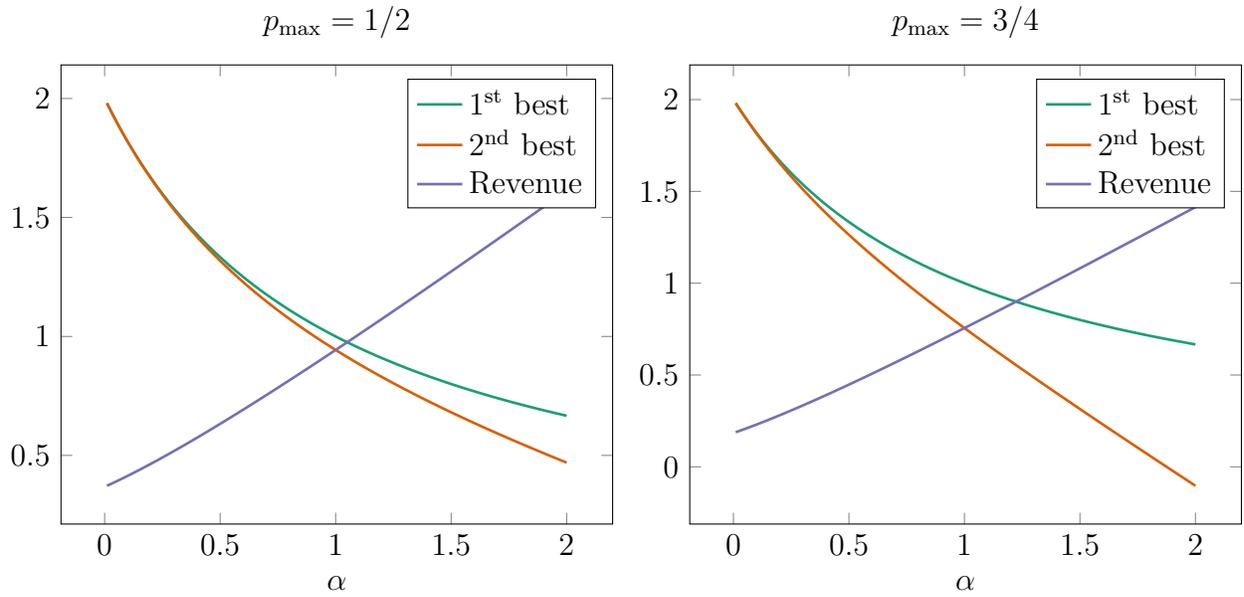

In Figure \ref{fig:opt-mech-uniform}, $v^\star<v_0$, downstream payoffs are non-decreasing as expected. When competition is more intense i.e. $\alpha=1/2$ the expected downstream payoff, $\tpi_i$, is smaller while payments are larger. This shows that the designer uses the competition between the firms as a tool to extract more surplus in exchange for the provision of exclusive information.

In Figure \ref{fig:opt-mech-nonuniform}, $v_0<v^\star$, the plots follow a similar pattern as Figure \ref{fig:opt-mech-uniform} except for types which are less than $v^\star$. For these types, fiercer competition results in lower payments. The increase in competition does not provide data buyers with exclusive information, and at the same time their expected downstream payoff decreases. Therefore, they do not have any incentive to pay more. However, Figure \ref{fig:wel/rev-alpha} suggests that the overall expected payment $\E\big[\tp_i(V_i)\big]$ increases with $\alpha$.

Figure \ref{fig:wel/rev-alpha} shows that welfare decreases as competition increases. By increasing $\alpha$ the entries of the payoff matrix decrease and as a result, the expected downstream payoff, $\tpi_i(V_i)$, and welfare decrease. On the other hand, revenue increases in a more competitive environment as the data seller can threaten each data buyer to provide exclusive information to their rivals.

\subsection{Number of Buyers}\label{sec:nplayers}

Another critical factor affecting the competition that a buyer faces in the downstream game is the total number of buyers. We now explore the impact of $n$ on the optimal mechanisms.

In contrast to previous sections where we focused on the case $n=2$ to visualize the impact of other factors, this requires considering the full generality of an $n$-dimensional type space, which is inherently hard to visualize. For this reason, we adopt the perspective of a single buyer $i\in[n]$ and study a two-dimensional slice of the type space parametrized by buyer $i$'s type $v_i$, and the sum $s_{-i} = \sum_{j\neq i} v_j$ of the other buyers' types. Conveniently, the recommendation to buyer $i$ in the welfare-optimal mechanism (\Cref{prop:second-best}) remains deterministic with this parametrization:
\begin{displaymath}
A_i=\theta\quad\text{if and only if}\quad
s_{-i}\leq\max\set{v^\star, (n-1)v_i/\alpha}.
\end{displaymath}
Similarly, buyer $i$'s recommendation in the revenue-optimal mechanism (\Cref{prop:revenue-alpha}) is deterministic in $\phi(v_i)$ and $s^{\phi}_{-i}\eqdef\sum_{j\neq i} \phi(v_j)$. In contrast, the externality induced by the recommendation to a buyer $j\neq i$ on buyer $i$ remains random even after conditioning on $v_i$ and $s_{-i}$ and depends on the conditional distribution of $v_j$ given $s_{-i}$.

Focusing on exponentially distributed type, $s_{-i}$ follows an Erlang distribution, for which the conditional distribution $v_j\given s_{-i}$ can be computed in closed-form. This allows us to visualize the externality induced by buyer $j$ on buyer $i$ by plotting the quantity $\P[A_j=\theta\given V_i, S_{-i}]$ as a function of the two parameters $(V_i, S_{-i})$. \Cref{fig:heatmaps} shows a heatmap of this function in the revenue-optimal mechanism for two different values of $n$.

\begin{figure}[!t]
	\includegraphics[height=0.38\textwidth]{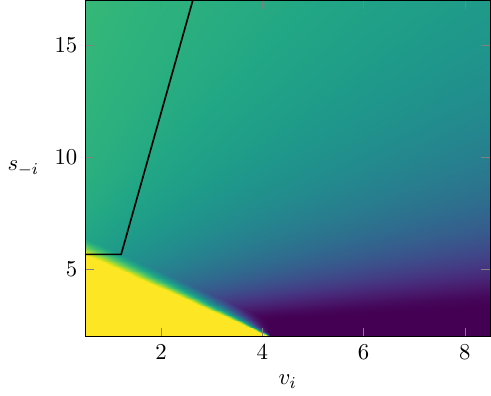}\hfill\includegraphics[height=0.395\textwidth]{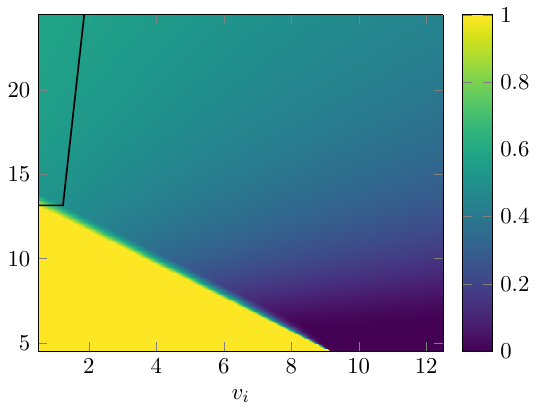}
	\caption{Heatmap of the function $(V_i, S_{-i})\mapsto \P[A_j=\theta\given V_i, S_{-i}]$ in the revenue-optimal mechanism for $n=5$ (left) and $n=10$ (right). Types are distributed as $1/2+Y$ where $Y$ is a standard exponential variable and $\alpha=\pmax=1/2$. The solid black line shows the boundary determining the recommendation to buyer $i$: $A_i=\theta$ below the line and $A_i=1-\theta$ above it.}
	\label{fig:heatmaps}
\end{figure}

Finally, we turn to the impact of the number $n$ of buyers on revenue and welfare. Specifically, we consider these objectives after normalization by $n$: by symmetry these correspond respectively to the utility (gross of any payments to the seller) and payment of a single buyer. As a way to assess the efficiency loss induced by the seller, we also consider the utility of a single buyer in the revenue-optimal mechanism. \Cref{fig:opt-mech-uniform-fb} shows how these quantities vary with $n$ for two different values of $\alpha$. As we see, all three quantities converge to a constant as $n$ grows to infinity. This can be heuristically explained as follows: because the externality term is normalized by $n-1$ in \eqref{eq:pi}, the optimal recommendation to buyer $i$, absent obedience constraints, is obtained by comparing their type $v_i$ to the \emph{average} of the other buyers' types. The law of large numbers then implies that in the large $n$ limit, the situation that buyer $i$ faces is identical to a competition with a \emph{single} buyer whose type is concentrated on the mean of the type distribution. Consistent with our findings from the previous section, we also see that as $\alpha$ increases---competition becomes fiercer---the payments increase while buyers' utilities decrease.

\usepgfplotslibrary{groupplots}
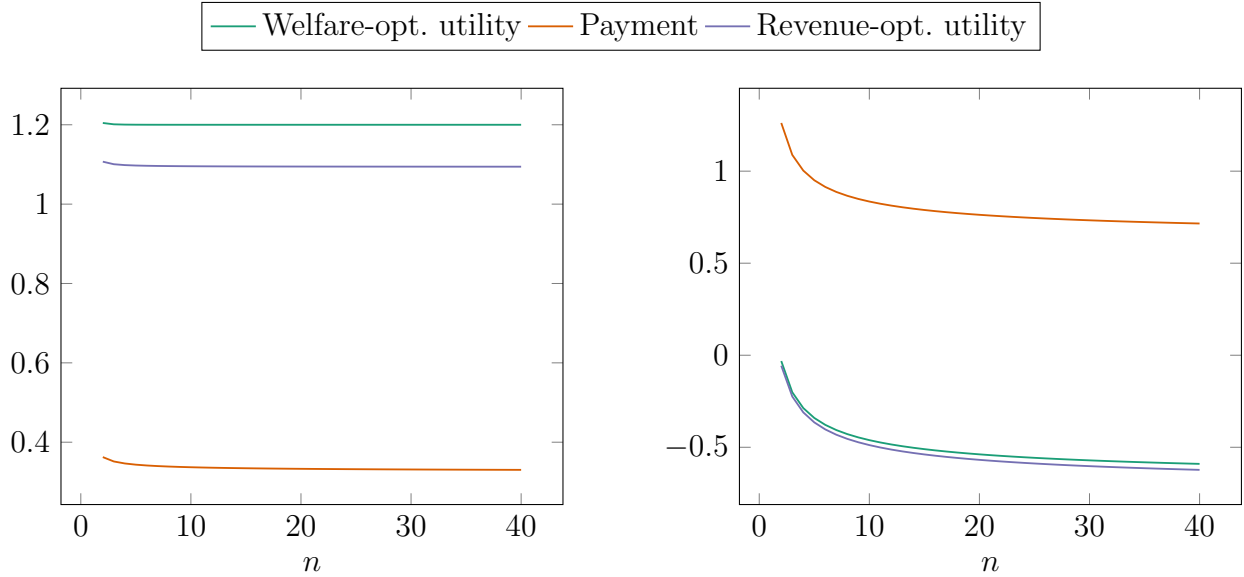
\begin{figure}
\begin{center}
\pgfplotslegendfromname{named}
\end{center}
\noindent\begin{tikzpicture}[baseline]
	\begin{axis}[
		xlabel={$n$},
		width=0.5\textwidth,
		cycle list/Dark2,
		legend cell align=left,
		legend entries={Welfare-opt.\ utility, Payment, Revenue-opt.\ utility},
		legend columns=-1,
		legend style={at={(0,1)}, anchor=south west},
		legend to name=named,
	]
		\addplot+[line width=0.7pt,mark=none] table [col sep=comma,x=n,y=Welfare] {WRWRon-alpha0.2-aon.dat};
		\addplot+[line width=0.7pt,mark=none] table [col sep=comma,x=n,y=Revenue] {WRWRon-alpha0.2-aon.dat};
		\addplot+[line width=0.7pt,mark=none] table [col sep=comma,x=n,y=Welfare at Revenue Optimal] {WRWRon-alpha0.2-aon.dat};
	\end{axis}
\end{tikzpicture}\hfill\begin{tikzpicture}[baseline]
	\begin{axis}[
		xlabel={$n$},
		width=0.5\textwidth,
		cycle list/Dark2,
	]
		\addplot+[line width=0.7pt,mark=none] table [col sep=comma,x=n,y=Welfare] {WRWRon-alpha2-aon.dat};
		\addplot+[line width=0.7pt,mark=none] table [col sep=comma,x=n,y=Revenue] {WRWRon-alpha2-aon.dat};
		\addplot+[line width=0.7pt,mark=none] table [col sep=comma,x=n,y=Welfare at Revenue Optimal] {WRWRon-alpha2-aon.dat};
	\end{axis}
\end{tikzpicture}\hfill
\caption{Utility of a single buyer in the welfare- and revenue-optimal mechanisms (green and purple) and payment in the revenue-optimal mechanism (orange) for $\alpha = 0.2$ (left), and $\alpha = 2$ (right). Types are distributed as $1/2+Y$ where $Y$ is a standard exponential variable and $\pmax=1/2$.}
\label{fig:opt-mech-uniform-fb}
\end{figure}


\section{Conclusions}
We have explored the implications of selling information to competing buyers in a mechanism design framework. The nature of information disciplines the optimal mechanisms for selling data products and distinguishes them from canonical (e.g., physical) goods. In particular, the buyers' actions in the downstream game introduce obedience constraints into the designer's choice of mechanism. These constraints prevent a social planner from  implementing the efficient degree of information \emph{exclusivity}---the second-best mechanism involves symmetric allocations of information more often than optimal. At the same time, obedience also limits the allocation distortions introduced by  a monopolist seller of information---the revenue-optimal mechanism provides the correct information to the buyers more often than the monopolist would like. 

In the present work, we characterized optimal mechanisms in the context of a linear model with binary states and actions. Considerable work remains to be done to extend the applicability of this framework. Natural next steps include introducing asymmetric buyers and externality parameters, which could be represented, for example, by a weighted directed graph. Both  extensions can be analyzed in our framework. Finally, in many real-world markets, information generates \emph{nonlinear} externalities that also depend on all buyers' actions in the downstream game. In these settings,  the sale of information to competing buyers creates value not only by allowing buyers to match their actions to the state but also by enabling coordination. In ongoing work \citep{bdhn22b}, we pursue the coordination role of selling \emph{Gaussian} information sales in a linear-quadratic  model.

\newpage

\appendix
\newpage

\section{Missing proofs from Section~\ref{sec:ic-char}}\label{sec:app-char}

\begin{proposition}[Proposition~\ref{prop:IC characterization} restated]
	Under Assumption~\ref{ass:lin}, a mechanism is incentive compatible whenever it is truthful and obedient.
\end{proposition}

\begin{proof}
	Consider a truthful and obedient mechanism. Then, using the notations of
	Definition~\ref{def:ic}, we have for each $(v_i,v_i')\in\V_i$:
	\begin{equation}\label{eq:foo-a1}
    \begin{split}
		\E\big[u_i(A;\theta,V)&-p_i(\theta,V)|V_i = v_i, B_i=v_i]\\
		&= \E\big[v_i\cdot\pi_i(A;\theta)-p_i(\theta,V)|V_i = v_i, B_i=v_i]\\
        &\geq \E\big[v_i\cdot\pi_i(A;\theta)-p_i(\theta,v_i',V_{-i})|V_i = v_i, B_i = v_i']\\
        &= \E\big[v_i\cdot\pi_i(A;\theta)-p_i(\theta,v_i',V_{-i})|B_i = v_i']\\
		&= v_i\cdot\E\big[\pi_i(A;\theta)\given B_i=v_i']
			-\E[p_i(\theta,v_i',V_{-i})|B_i = v_i'],
    \end{split}
\end{equation}
	where the first equality is by Assumption~\ref{ass:lin}, the inequality is by truthfulness (Definition~\ref{def:truthful}), and the third equality is because $A$ is independent of $V_i$ conditioned on $B_i$ by Assumption~\ref{ass:lin}.

	Let $\delta:A_i\to A_i$ be a deviation function. By Assumption~\ref{ass:lin}, obedience
	(\Cref{def:obedience}) is equivalent to 
	\begin{displaymath}
		\E[\pi_i(A;\theta)\given B_i=v_i]\geq \E[\pi_i(\delta(A_i),
		A_{-i};\theta)\given B_i=v_i]
	\end{displaymath}
	for all $v_i\in\V_i$. Applying this inequality for $v_i'$ in \eqref{eq:foo-a1} we
	obtain
	\begin{align*}
		\E\big[u_i(A;\theta,V)-p_i(\theta,V)|V_i = v_i, B_i=v_i]
			&\geq v_i\cdot\E\big[\pi_i(\delta(A_i), A_{-i},\theta)\given
				B_i=v_i']\\
			&\quad-\E[p_i(\theta,v_i',V_{-i})|B_i = v_i']\\
			&\geq \E\big[u_i(\delta(A_i), A_{-i};\theta, V)\given V_i=v_i,
				B_i=v_i']\\
			&\quad-\E[p_i(\theta,v_i',V_{-i})|V_i=v_i, B_i = v_i'],
	\end{align*}
	which is precisely the definition of incentive compatibility. 
\end{proof}

\begin{proposition}[\Cref{prop:expected-payment} restated]
	The mechanism $(\sigma, p)$ is truthful if and only if for each buyer $i$:
	\begin{enumerate}
		\item The interim downstream payoff $\tpi_i$ is non-decreasing.
		\item The interim payment $\tp_i$ satisfies, for all $v_i\in\V_i$,
			\begin{equation}\label{eq:truthful-payment-appendix}
				\tp_i(v_i) = v_i\cdot\tpi_i(v_i) - \uv\cdot\tpi_i(\uv)
				+ \tp_i(\uv) - \int_{\uv}^{v_i} \tpi_i(s)ds\,.
			\end{equation}
	\end{enumerate}
\end{proposition}

\begin{proof}
	Define $\tu_i:v_i\mapsto v_i\cdot \tpi_i(v_i) - \tp_i(v_i)$.
	Then truthfulness is equivalent to
	\begin{displaymath}
		\tu_i(v_i) = v_i\cdot \tpi_i(v_i) - \tp_i(v_i)
		\geq v_i\cdot\tpi_i(v_i') - \tp_i(v_i')
		= \tu_i(v_i') + (v_i-v_i')\cdot\tpi_i(v_i')\,,
	\end{displaymath}
	for all $(v_i, v_i')\in\V_i^2$. This is equivalent to saying that
	$\tpi_i(v_i)\in\partial \tu_i(v_i)$ for all $v_i\in\V_i$ where
	$\partial\tu_i(v_i)\subset\R$ denotes the subdifferential of $\tu_i$ at
	$v_i$. By a well-known characterization of convexity, this in turn
	equivalent to saying that $\tpi_i$ is non-decreasing and that 
	\begin{displaymath}
		\tilde u_i(v_i) = \tilde u_i(\uv) + \int_{\uv}^{v_i} \tpi_i(s)ds\,.
	\end{displaymath}
	This concludes the proof since this last expression is equivalent to
	\eqref{eq:truthful-payment-appendix}.
\end{proof}

\begin{proof}[Proof of Proposition \ref{lemm:char}]
	Since there are two states, obedience states that for all $i\in[n]$ and all $a_i\in\Theta$, 
\begin{displaymath}
	\E[\pi_i(a_i, A_{-i};\theta)-\pi_i(1-a_i, A_{-i};\theta)\given A_i=a_i,
	V_i]\geq 0.
\end{displaymath}
Using the form of $\pi_i$ from \Cref{ass:ext}, we observe that the externality
terms cancel out and the previous inequality is equivalent to
\begin{displaymath}
	\P[\theta=a_i\given A_i=a_i, V_i] \geq \P[\theta=1-a_i\given A_i=a_i, V_i],
\end{displaymath}
and by Bayes' rule to
\begin{displaymath}
	\P[\theta=a_i\land A_i=a_i\given V_i] \geq \P[\theta=1-a_i\land A_i=a_i\given V_i]
	\,.
\end{displaymath}
Adding the quantity $\P[\theta=1-a_i\land A_i=1-a_i\given V_i]$ on both sides,
we obtain
\begin{displaymath}
	\P[A_i=\theta\given V_i] \geq \P[\theta=1-a_i\given V_i]=\P[\theta=1-a_i]
	\,,
\end{displaymath}
where the last equality uses that $\theta$ and $V_i$ are independent by
Assumption~\ref{ass:lin}. Taking a maximum over $a_i\in\Theta$ yields the
lemma's statement.
\end{proof}

\section{Missing proofs from Section~\ref{sec:meta}}\label{sec:app-meta}


\begin{lemma}[Recommendation Rule from Marginals]\label{lemm:param}
Let $h_i$ be measurable functions from $\V$ to $[0,1]$ for $i\in[n]$, then there exists a recommendation rule $\sigma:\Theta\times\V\to\Delta(\A)$ such that almost surely,
$\P[A_i=\theta\given V] = h_i(V)$ for $i\in[n]$.
\end{lemma}

\begin{proof}
	Given functions $h_i$ satisfying the lemma's assumptions, one can choose
	for example $\sigma$ such that for all $(x,v)\in\Theta\times \V$, the
	distribution $\sigma(x,v)\in\Delta(\A)$ has independent coordinates with
	marginals given by $h_i$. Formally, we have for $(x,v)\in\Theta\times\V$
	and $a\in\A$,
	\begin{displaymath}
		\sigma(a;x,v) = \prod_{\mathclap{i:a_i=x}}
		h_i(v)\prod_{\mathclap{i:a_i\neq x}} \big(1-h_i(v)\big).\qedhere
	\end{displaymath}
\end{proof}

Before we prove \Cref{prop:master} we first state and prove a variational lemma
that solves the pointwise optimization problem that the optimal mechanism
reduces to.

\begin{lemma}[Variational Lemma]\label{lemm:inc}
	Let $(E,\mu)$ be a probability space and let $g:E\to\R$ be
	a $\mu$-integrable function whose level sets are $\mu$-null sets:
	$\mu(\set{v\in E\given g(v)=k})=0$ for all $k\in\R$. Consider the
	problem
	\begin{displaymath}
		\begin{aligned}
			\max_{h\in\F}&\;\;\L(h)\eqdef\int_{E} h\cdot g \,\mathrm{d}\mu\\
			\text{s.t.}&\;\int_{E} h \,\mathrm{d}\mu \geq c
		\end{aligned}
	\end{displaymath}
	where the optimization is over the set $\F$ of measurable functions
	$h:E\to\R$ with $h(E)\subseteq [0,1]$ and $c\in[0,1]$ is a constant. For
	$k\in\R$, define $L_g^+(k)\eqdef\set{v\in E\given g(v)\geq k}$ the
	superlevel set of $g$ of level $k$ and let $t_c \eqdef \sup\set{k\in\R\given
	\mu\big(L_g^+(k)\big)\geq c}$. Then, an optimal solution to the problem is
	given by $h^\star:v\mapsto\ind\set{g(v)\geq t^\star}$ where $t^\star
	= \min\set{t_c,0}$.
\end{lemma}

\begin{proof}
	First, note that $t_c$ is well-defined in the extended real line by adopting
	the usual convention $\sup\emptyset=-\infty$ and $\sup\R=+\infty$. If
	$k\leq k'$, we have $L_{k'}^+(g)\subseteq L_k^+(g)$ showing that the
	function $m:k\mapsto \mu\big(L_k^+(g)\big)$ is non-increasing. The identity
	$L_g^+(k)= \bigcap_{k'<k} L_g^+(k')$ implies
	\begin{displaymath}
		m(k) = \mu\big(L_g^+(k)\big) = \inf_{k'<k}\mu\big(L_g^+(k')\big)
		=\inf_{k'<k}m(k')
	\end{displaymath}
	and $m$ is thus left-continuous. Consequently, the threshold $t_c$ is
	characterized by the equivalence
	\begin{displaymath}
		k\leq t_c \Longleftrightarrow \mu\big(L_g^+(k)\big)\geq c.
	\end{displaymath}
	Furthermore, the identity $\set{v\in E\given g(v)>k}= \bigcup_{k'>k}
	L_{k'}^+(g)$ implies
	\begin{displaymath}
		\mu\big(\set{v\in E\given g(v)>k}\big)= \sup_{k<k'}\mu\big(L_g^+(k')\big)
		=\sup_{k<k'}m(k').
	\end{displaymath}
	But because the level sets of $g$ are $\mu$-null sets, we also have $m(k)
	= \mu\big(\set{v\in E\given g(v)>k}\big)$, hence the function $m$ is
	continuous, and the threshold $t_c$ further satisfies
	$\mu\big(L_g^+(t_c)\big)=c$. This implies that $h^\star$ is a feasible
	solution, indeed since $t^\star\leq t_c$ we have
	\begin{displaymath}
		\int_E h^\star\,\mathrm{d}\mu = \mu\big(L_g^+(t^\star)\big)
		\geq\mu\big(L_g^+(t_c)\big)=c.
	\end{displaymath}

	Next, for a feasible $h\in\F$ we have
	\begin{equation}\label{eq:fooo}
			\begin{split}
			\L(h^\star) - \L(h)
			&= \int_{L_g^+(t^\star)}(1-h)g\,\mathrm{d}\mu
			+\int_{E\setminus L_g^+(t^\star)} (-h)g\,\mathrm{d}\mu\\
			&\geq t^\star\int_{L_g^+(t^\star)}(1-h)\,\mathrm{d}\mu
			+t^\star\int_{E\setminus L_g^+(t^\star)} (-h)\,\mathrm{d}\mu\\
			&=t^\star\mu\big(L_g^+(t^\star)\big) - t^\star\int_E h\,\mathrm{d}\mu
			\geq t^*\Big[\mu\big(L_g^+(t^\star)\big)-c\Big],
			\end{split}
		\end{equation}
		where the first equality is by definition of $h^\star$, the subsequent
		equality uses that $h$ takes values in $[0,1]$ and the fact that
		$g(v)\geq t^\star$ iff $v\in L_g^+(t^\star)$, and the last inequality
		uses that $t^\star\leq 0$ by definition and that $\int_E
		h\,\mathrm{d}\mu\geq c$ by feasibility.

		Either $t^\star = 0$, or $t^\star = t_c$ in which case
		$\mu\big(L_g^+(t^\star)\big) = c$. In both cases, the last expression
		in \eqref{eq:fooo} vanishes, hence $\L(h^\star)\geq \L(h)$ for
		all feasible $h\in\F$ which concludes the proof.
\end{proof}

We can now describe the optimal mechanism.
\begin{proposition}[Prop.~\ref{prop:master} restated]
Under Assumptions~\ref{ass:lin} and \ref{ass:ext}, consider an objective $W$ of the form \eqref{eq:meta} where for $i\in[n]$, $w_i:\V\to\R$ is a measurable function such that the random variable $w_i(v_i, V_{-i})$ is nonatomic for each $v_i\in\V_i$. For $i\in[n]$ let $t_i^\star:\V_i\to\R$ be such that for all $v_i\in\V_i$,
\begin{displaymath}
	\P\big[w_i(v_i, V_{-i})\geq t_i^*(v_i)\big]=\max_{k\in\set{0,1}}\P[\theta=k]=:\pmax.
\end{displaymath}
Then the deterministic recommendation rule given by
\begin{displaymath}
A_i=\theta\quad\text{if and only if}\quad
w_i(v)\geq \min\set{0,t_i^\star(v_i)}
\end{displaymath}
for $i\in[n]$, maximizes $W$ subject to obedience.
\end{proposition}

\begin{proof}
First note that since $w_i(v_i, V_{-i})$ is non-atomic, its c.d.f. is continuous for all $v_i\in\V_i$. Hence, defining $c\eqdef \max_{k\in\set{0,1}}\P[\theta=k]$, a suitable threshold function $t_i^\star$ can be obtained by choosing for $v_i\in\V_i$,
\begin{displaymath}
t_i^\star(v_i) \eqdef \sup\set[\big]{k\in\R\given \P[w_i(v_i, V_{-i})\geq k]\geq c}.
\end{displaymath}
Define $h_i:\V\to[0,1]$ by $h_i(V)=\P[A_i=\theta\given V]$. By \Cref{lemm:char} and using the law of iterated expectations, the obedience constraint for buyer $i\in[n]$ states that
\begin{displaymath}
\P[A_i=\theta\given V_i] = \E[\ind\set{A_i=\theta}\given V_i]
	= \E\big[\E[\ind\set{A_i=\theta}\given V]\given
	V_i\big]=\E[h_i(V)\given V_i]\geq c\,
\end{displaymath}
almost surely for $V_i$. Similarly we can rewrite the objective \eqref{eq:meta} in terms of the functions $h_i$ and the optimization problem we need to solve can thus be written
\begin{align*}
	\max&\;\sum_{i=1}^n\E\big[w_i(V)h_i(V)\big]\\
	\text{s.t.}&\;\E[h_i(V)\given V_i]\geq c,\;\text{for $i\in[n].$}
\end{align*}

Because both the objective function and the constraints are separable in $i$,
this problem decomposes as $n$ separate optimization problem, one for
each $h_i$, $i\in[n]$:
\begin{align*}
	\max&\;\E\big[\E[w_i(V_i, V_{-i})h_i(V_i, V_{-i})\given V_i]\big]\\
	\text{s.t.}&\;\E[h_i(v_i, V_{-i})]\geq c,\;\text{for all $v_i\in\V_i$}
\end{align*}
where we used the law of total expectation and the independence of $(V_1,\dots,V_n)$. Finally, since the constraint is a pointwise constraint for $v_i\in\V_i$, the optimal $h_i$ is obtained by choosing the partial function $h_i(v_i,\cdot)$ so as to maximize the integrand in the objective function for each $v_i$. That is, $h_i(v_i,\cdot)$ should solve
\begin{align*}
	\max&\;\E\big[w_i(v_i, V_{-i})h_i(v_i, V_{-i})\big]\\
	\text{s.t.}&\;\E[h_i(v_i, V_{-i})]\geq c.
\end{align*}
This problem is exactly of the form solved in \Cref{lemm:inc}, with $\mu$ being the probability distribution of $V_{-i}$ and with $g = w_i(v_i,\cdot)$ and $E=\R^{n-1}$. By construction, the threshold $t_i^\star(v_i)$ plays the role of $t_c$ in \Cref{lemm:inc} for this choice of function $g$. Hence, the optimal policy is given by
\begin{displaymath}
h_i(v_i, v_{-i}) = \ind\set[\big]{w_i(v_i, v_{-i})\geq\min\set{0, t_i^\star(v_i)}}\text{ for each $v_i\in\V_i$ },
\end{displaymath}
which is the deterministic rule in the proposition statement.
\end{proof}

\section{Missing proofs from Section~\ref{sec:welfare}}\label{sec:app-welfare}

\begin{proposition}[Prop.~\ref{prop:second-best} restated]
Consider the binary game with additive payoffs \eqref{eq:pi} under
\Cref{ass:lin}. Further assume that the buyers' types are identically
distributed with absolutely continuous c.d.f.\ $F$ and denote by $F^{(k)}$ the
c.d.f.\  of the sum of $k$ i.i.d.\ variables\footnote{$F^{(k)}$ can be computed
recursively with $F^{(1)} = F$ and $F^{(k+1)} = F^{(k)}\ast f$, where $\ast$
denotes the convolution product and $f$ is the p.d.f.\ associated with $F$.}
distributed according to $F$. Define $v^\star$ such that $F^{(n-1)}(v^\star) =
\pmax$ and $\overline\alpha\eqdef\frac\alpha{n-1}$.

Then, the recommendation rule maximizing social welfare subject to obedience is the deterministic rule given by
\begin{displaymath}
A_i=\theta\quad\text{if and only if}\quad
\sum_{j\neq i} v_j\leq\max\set{v^\star, v_i/\overline\alpha}.
\end{displaymath}
\end{proposition}

\begin{proof}[Proof of Proposition \ref{prop:second-best}]
For the payoffs described by \eqref{eq:pi}, \Cref{ass:ext} holds and the expected welfare~\eqref{eq:welfare} is of the form \eqref{eq:meta} covered by \Cref{prop:master}, with weight function $w_i(v) = v_i-\overline\alpha\sum_{j\neq i} v_j$. Furthermore, the level sets of the partial function $w_i(v_i,\cdot)$ are hyperplanes in dimension $n$ and since the distribution of $V_{-i}$ is absolutely continuous by assumption (with c.d.f.\ $F^{(n-1)}$), the random variable $w_i(v_i, V_{-i})$ is non-atomic for each $v_i\in\V_i$. Hence, \Cref{prop:master} applies and the recommendation rule maximizing welfare subject to obedience is determined by
\begin{displaymath}
A_i=\theta\quad\text{if and only if}\quad
w_i(v)\geq \min\set{0,t_i^\star(v_i)}.
\end{displaymath}

All that remains is to compute the threshold function $t_i^\star(v_i)$. By
definition we must have
\begin{displaymath}
	\P\big[w_i(v_i, V_{-i})\geq t_i^\star(v_i)\big]
	=\Pr\left[\sum_{j\neq i} V_j \leq \frac{v_i-t_i^\star(v_i)}{\overline\alpha}\right]
	=F^{(n-1)}\left(\frac{v_i-t_i^\star(v_i)}{\overline\alpha}\right)
	=\pmax,
\end{displaymath}
which by definition of $v^\star$ is equivalent to
$$\frac{v_i-t_i^\star(v_i)}{\overline\alpha} = v^\star.$$ The result in the statement follows after  observing that
$w_i(v)\geq\min\set{0,t_i^\star(v_i)}$ is equivalent to $\sum_{j\neq
i}v_j\leq\max\set{v^\star, v_i/\overline\alpha}$.
\end{proof}

\begin{proof}[Proof of Proposition \ref{prop:implementability}]
Using the definition of the optimal recommendation rule and
	linearity of conditional expectations we have
	\begin{align*}
		\tpi_i(V_i)
		&= \E\Big[\ind\set{A_i=\theta}-\overline\alpha\sum_{j\neq
	i}\ind\set{A_j=\theta}\given[\Big] V_i \Big]\\
		&= \E\bigg[\ind\set[\bigg]{\sum_{j\neq i}V_j\leq \max\set{v^\star,
V_i/\overline\alpha}}\given[\bigg] V_i \bigg]
		-\overline\alpha\sum_{j\neq i}\E\bigg[\ind\set[\bigg]{\sum_{k\neq j}V_k\leq
\max\set{v^\star, V_j/\overline\alpha}}\given[\bigg] V_i \bigg].
\end{align*}
The first expectation on the right-hand side can be computed as
\begin{displaymath}
\E\bigg[\ind\set[\bigg]{\sum_{j\neq i}V_j\leq \max\set{v^\star,
V_i/\overline\alpha}}\given[\bigg] V_i \bigg]
= \max\set[\big]{F^{(n-1)}(v^\star),F^{(n-1)}(V_i/\overline\alpha)}. 
\end{displaymath}
For $j\neq i$, the summand expectation can be computed using the law of total
expectations as
\begin{align*}
	\E\bigg[\ind\set[\bigg]{\sum_{k\neq j}V_k\leq
\max\set{v^\star, V_j/\overline\alpha}}\given[\bigg] V_i \bigg]
&= \E\Bigg[\E\bigg[\ind\set[\bigg]{\sum_{k\neq j}V_k\leq
\max\set{v^\star,V_j/\overline\alpha}}\given[\bigg] V_i, V_j \bigg]\given[\bigg]V_i\Bigg]\\
&= \E\Big[F^{(n-2)}\big(\max\set{v^\star,V_j/\overline\alpha}-V_i\big)\given[\Big]V_i\Big].
\end{align*}

The claim that $\tpi_i$ is a non-decreasing function easily follows from the observation that $v\mapsto\max\set[\big]{F^{(n-1)}(v^\star),F^{(n-1)}(v/\overline\alpha)}$ is non-decreasing (because $F^{(n-1)}$ is non-decreasing as a c.d.f.) and that for all $j\neq i$, the quantity $\E\Big[F^{(n-2)}\big(\max\set{v^\star,V_j/\overline\alpha}-v\big)\Big]$, is non-increasing as a function of $v$ because the integrand is non-increasing in $v$ pointwise.
\end{proof}

\section{Missing proofs from Section~\ref{sec:revenue}}\label{sec:app-revenue}

\begin{lemma}[\Cref{lem:vv} restated]
Let $\sigma$ be a communication rule for which the interim payoff $\tpi_i$ is non-decreasing for each buyer $i\in[n]$. Denote by $K$ the interim payoff of a non-participating buyer in their outside option and assume that $\tpi_i(\uv)\geq K$. Then:
	\begin{enumerate}
		\item if $p_i$ is a payment function that truthfully implements
			$\tpi_i$ (i.e., that satisfies \eqref{eq:truthful-payment-appendix}
			by \Cref{prop:expected-payment}—then $(\sigma,p)$ is individually
			rational if and only if it is individually rational at the lowest
			type, that is, $\tp_i(\uv)\leq \uv\cdot (\tpi_i(\uv) - K)$.
		\item among the payment functions $p_i$ implementing $\tpi_i$ in
			a truthful and individually rational manner, the revenue-maximizing
			one is given by
		\begin{equation}\label{eq:payments-restated}
			\tp_i(v_i) = v_i\cdot\tpi_i(v_i) - \uv\cdot K - \int_{\uv}^{v_i}\tpi_i(s)ds\,.
		\end{equation}
		For this payment function, the seller's revenue is $R=
		\sum_{i=1}^n\E\big[\phi(V_i)\tpi_i(V_i)\big]-n\uv\cdot K$.
	\end{enumerate}
\end{lemma}

\begin{proof}
	\begin{enumerate}
		\item Assume that $p_i$ truthfully implements $\tpi_i$, so that the
			interim utility $\tu_i$ is convex by \Cref{prop:expected-payment}.
			If $(\sigma, p)$ is individually rational, then it is obviously
			individually rational at the lowest type. For the converse
			direction, assume that $\tu_i(\uv)\geq \uv\cdot K$, then we have
			for all $v_i\in\V_i$
	\begin{displaymath}
		\tu_i(v_i)\geq \tu_i(\uv) + (v_i-\uv)\cdot\tpi_i(\uv)\geq v_i\cdot K\,,
	\end{displaymath}
	where the first inequality uses convexity of $\tu_i$ and
	the second inequality uses that $\tpi_i(\uv)\geq K$ and individual
	rationality at the lowest type.

\item If $p_i$ implements $\tpi_i$ truthfully, then we already know by
	\Cref{prop:expected-payment} that it satisfies
	\eqref{eq:truthful-payment-appendix}, in which the only undetermined
	quantity is the interim payment at the lowest type $\tp_i(\uv)$. By 1., 
	individual rationality further constrains $\tp_i(\uv)\leq
	\uv\cdot(\tpi_i(\uv) - K)$, so the revenue-maximizing choice makes this
	constraint bind and we get \eqref{eq:payments-restated}.

	For this payment function, we write the seller's revenue as
\begin{align}
	R &= \E\Big[\sum_{i=1}^np_i(\theta,V)\Big] = \E\Big[\sum_{i=1}^n\E[p_i(\theta,V)
	\given V_i] \Big]=\E\Big[\sum_{i=1}^n \tp_i(V_i) \Big]\notag\\
	  &= \sum_{i=1}^n\E\bigg[V_i\cdot\tpi_i(V_i) - \uv\cdot K
		- \int_{\uv}^{V_i} \tpi_i(s)ds\bigg]\notag \\
	&= \sum_i\E\bigg[V_i\cdot\tpi_i(V_i) - \int_{\uv}^{V_i}
	\tpi_i(s)ds\bigg] -n\uv\cdot K.\label{eq:rev}
\end{align}
The expectation in the first summand is computed as follows
\begin{align*}
    \E\Big[V_i\cdot\tpi_i(V_i) - \int_{\uv}^{V_i} \tpi_i(s)ds\Big]
    & = \int_{\uv}^{\overline{v}}v_i\cdot\tpi_i(v_i)f(v_i)dv_i -\int_{\uv}^{\overline{v}} \int_{\uv}^{v_i} \tpi_i(s)f(v_i)dsdv_i\\
    & = \int_{\uv}^{\overline{v}}v_i\cdot\tpi_i(v_i)f(v_i)dv_i -\int_{\uv}^{\overline{v}} \int_{s}^{\overline{v}} \tpi_i(s)f(v_i)dv_ids\\
    & = \int_{\uv}^{\overline{v}}\left(v_i-\frac{1-F(v_i)}{f(v_i)}\right)\cdot\tpi_i(v_i)f(v_i)dv_i\\
    & = \E\big[\phi(V_i)\tpi_i(V_i)\big]\,,
\end{align*}
where the second equality uses Fubini's theorem and the last equality is the
definition of the virtual value function.\qedhere
\end{enumerate}
\end{proof}

\begin{proof}[Proof of Proposition \ref{prop:rev-implementability}]
The computation of the interim payoff $\tpi_i$ is identical to the one in the proof of \Cref{prop:implementability} after replacing the types with their virtual counterparts.

It is immediate from this expression that for all $v_i\in\V_i$,
\begin{displaymath}
	\tpi_i(v_i) \geq  F_\phi^{(n-1)}\big(\phi(v^\star)\big) - \alpha
	= \pmax - \alpha,
\end{displaymath}
where the inequality is obtained by upper-bounding $F_\phi^{(n-2)}$ by $1$ and
the equality is by definition of $v^\star$. Furthermore, for the outside
allocation described in the proposition statement, the best response of buyer
$i$ in case of non-participation is to play the action matching the most likely
state under the prior, resulting in an expected payoff
	$K \eqdef \pmax - \alpha$.

Indeed, the non-participating buyer will be correct with probability $\pmax$ while incurring an externality of $-\alpha$ from all the participating buyers (who then receive the correct action recommendation). The previous two equations combined show that $\tpi_i(v_i)\geq K$ and the conditions of \Cref{lem:vv} are thus satisfied. The revenue-maximizing payments subject to truthfulness and individual rationality are then given by \cref{eq:payments-bis}.
\end{proof}

\newpage
\bibliographystyle{apalike}
	\bibliography{main}

\end{document}